\documentclass[10pt,a4paper]{article}
\usepackage{amssymb}
\usepackage{amsmath}
\usepackage{amsthm}
\usepackage{latexsym}
\usepackage[dvips]{epsfig}
\usepackage{mathrsfs}
\usepackage{eufrak}
\usepackage{bm}
\usepackage{stmaryrd}
\usepackage{authblk}

\usepackage{tikz}

\theoremstyle{plain}
\newtheorem{proposition}{Proposition}
\newtheorem{lemma}{Lemma}
\newtheorem{theorem}{Theorem}

\newtheorem*{main}{Theorem}

\font\tenscr=rsfs10 scaled1100
\font\sevenscr=rsfs7 
\font\fivescr=rsfs5 
\skewchar\tenscr='177
\skewchar\sevenscr='177
\skewchar\fivescr='177
\newfam\scrfam
\textfont\scrfam=\tenscr
\scriptfont\scrfam=\sevenscr
\scriptscriptfont\scrfam=\fivescr

\newcommand{\updn}[3]{#1^{#2}_{\phantom{#2}#3}}

\newcounter{mnote}

\begin{document}


\title{\textbf{A class of conformal curves in the Reissner-Nordstr\"om spacetime}}

\author[,1,2]{Christian L\"ubbe \footnote{E-mail address:{\tt c.luebbe@qmul.ac.uk}, {\tt c.luebbe@ucl.ac.uk}}}
\author[,2]{Juan Antonio Valiente Kroon \footnote{E-mail address:{\tt j.a.valiente-kroon@qmul.ac.uk}}}

\affil[1]{Department of Mathematics, University College London, Gower
  Street, London WC1E 6BT, UK}

\affil[2]{School of Mathematical Sciences, Queen Mary, University of London,
Mile End Road, London E1 4NS, United Kingdom}

\maketitle

\begin{abstract}
A class of curves with special conformal properties (conformal curves)
is studied on the Reissner-Nordstr\"om spacetime. It is shown that
initial data for the conformal curves can be prescribed so that the
resulting congruence of curves extends
smoothly to future and past null infinity. The formation of conjugate
points on these congruences is examined.  The results of this
analysis are expected to be of relevance for the discussion of the
Reissner-Nordstr\"om spacetime as a solution to the conformal field equations
and for the global numerical evaluation of
static black hole spacetimes.
\end{abstract}

PACS: 04.20.Ha, 04.20.Jb, 04.70.Bw

Keywords: black holes, conformal structure

\section{Introduction}

Conformal methods constitute a powerful tool for the discussion of
global properties of spacetimes ---in particular, those representing
black holes. The conformal structure of static electrovacuum black
hole spacetimes is, to some extent, well understood ---see
e.g. \cite{GriPod09,HawEll73}. However, the constructions involved 
often require several changes of variables and the
introduction of some type of null coordinates. This choice of
coordinates may not be the most convenient to undertake an analysis of
global or asymptotic properties of a spacetime by means of the
\emph{conformal Einstein field equations} ---see e.g. the discussion in
\cite{Fri98a}. A key issue in this respect, is 
how to construct in a systematic/canonical fashion a conformal
extension of the spacetime which, in addition, eases the analysis of
the underlying conformal field equations ---for a review of the
conformal equations and the issues involved in their analysis see
e.g. \cite{Fri03a}. In the case of vacuum spacetimes, gauges based on the use of \emph{conformal
geodesics} offer such a systematic approach ---see
e.g. \cite{Fri03c,FriSch87}. Conformal geodesics are invariants of the
conformal structure: a conformal transformation maps  conformal
geodesics into conformal geodesics  ---this is not the case with
standard geodesics unless they are null.

\medskip
One of the main advantages of the
use of conformal geodesics in the construction of gauge (and
coordinate) systems in a vacuum spacetime is that they provide an \emph{a priori} conformal
factor which can be read off directly from the data one has specified
to generate the congruence of conformal geodesics. Hence, one has a
canonical procedure to generate a conformal extension of the spacetime
in question.  In addition, gauge systems based on conformal geodesics
give rise to a fairly straightforward hyperbolic reduction of the
conformal Einstein field equations in which most of the evolution
equations are, in fact, transport equations ---see
e.g. \cite{Fri03a,Fri04}.

\medskip
The useful property of having an \emph{a priori} conformal factor is
lost when one considers conformal geodesics in non-vacuum
spacetimes. Nevertheless, in \cite{LueVal12} it has been shown that
this property can be recovered if one considers a more general class
of curves ---the \emph{conformal curves}. These curves satisfy
equations similar to the conformal geodesic
equations, but with a different coupling to the curvature of the
spacetime. In the vacuum case they coincide with the conformal
geodesic equations. Gauges based on this class of curves have been used in \cite{LueVal12} to
revisit the stability proofs for the Minkowski and the de Sitter
spacetimes first given in \cite{Fri91} and to obtain a stability
result for purely radiative electrovacuum spacetimes. They also have
been used in \cite{LueVal13a} to analyse the geodesic completeness of
non-linear perturbations of Friedman-Robertson-Walker spacetimes with
radiation perfect fluids.

\medskip
Given the results described in the previous paragraph, a natural
question to be raised is  whether conformal geodesics, and more
generally, the class of conformal curves introduced in
\cite{LueVal12} can be used to analyse global aspects of black hole
spacetimes. A first analysis of this question has been carried out in
\cite{Fri03c} where it has been shown that
the maximal extension of the Schwarzschild spacetime, the so-called
Schwarzschild-Kruskal spacetime \cite{Kru60}, can be covered with a congruence of conformal
geodesics which has no conjugate points. The conformal Gaussian gauge
system obtained using this congruence offers a vantage perspective for
the study of conformal properties of the Schwarzschild spacetime and
for its global evaluation by means of numerical methods ---see e.g. \cite{Zen06}.

\medskip
In the present article we analyse to what extent a similar
construction can be performed for the Reissner-Nordstr\"om
spacetime. The idea of considering the Reissner-Nordstr\"om spacetime
is, for several reasons, natural. The inclusion of the electromagnetic
field provides a model of angular momentum ---see
e.g. \cite{Daf03,Daf05}. Moreover, we expect our analysis to provide
insights into more general (i.e. less symmetric) situations
---e.g. the Kerr and Kerr-Newman spacetimes. In addition, there is the
expectation that black hole spacetimes with timelike singularities
could be more tractable from the point of view of the conformal
geometry than black holes with spacelike singularities\footnote{This
expectation is based on the analysis of the structure of spatial
infinity of the Schwarzschild spacetime. In this case the well
understood divergence of the Weyl tensor at spatial infinity can also
be regarded as the timelike singularity of a negative mass
Schwarzschild spacetime ---see e.g. \cite{SchWal83} for a conformal
diagram of this.}.

\medskip
The main results of our analysis is the following: 

\begin{main}
The domain of outer communication of a non-extremal
Reissner-Nordstr\"om spacetime with $q^2\leq \tfrac{8}{9}m^2$ can be
covered with a timelike congruence of conformal curves which contains
no conjugate points. In the extremal case $m^2=q^2$ the
non-existence of conjugate points for an analogous congruence can
be ensured, in
the worst of cases, for the region in the domain of outer communication outside a certain timelike
tube intersecting the horizon. In
both cases, the congruence of conformal curves extends smoothly to
null infinity.
\end{main}

A technical version of the main result is given in Theorem
\ref{MainTheoremTechnicalVersion} of the conclusions.

\medskip
Numerical solutions of the conformal curve equations show that
for $\tfrac{8}{9}m^2 < q^2 \leq m^2$ the congruence of conformal
curves contains no conjugate points in the domain of outer
communication. Thus, our main result can be certainly be
improve. Doing this, however, would increase considerably the length
of our analysis. In view of future applications, the extremal case is
certainly the one of the most interest. The same numerical simulations
show that, generically, conjugate points in the congruence form after
the curves have crossed the horizon and entered the black hole region
of the spacetime. From the perspective of the
Cauchy problem for the Einstein field equations, these conjugate
points are not a major concern as one is mainly interested in the
behaviour of the spacetime in the domain of outer communication and at
the horizon. This is, in particular, the case in the problem of the
so-called \emph{non-linear stability of black hole spacetimes} ---see
e.g. \cite{DafRod08}. 

\medskip
 Our
main result provides a suitable conformal gauge to analyse the
properties of the Reissner-Nordstr\"om spacetime by means of the
conformal Einstein field equations. In particular, it opens the
possibility of global numerical evaluations of the spacetime
\cite{Val12} similar to the ones carried out in \cite{Zen06} for the
Schwarzschild spacetime.

\medskip
Finally, we point out that some interesting recent work on other aspects of the Reissner-Nordstr\"om
spacetime can be found in \cite{Are11a,Are11b,DaiDot12,BizFri12}.

\subsection*{Outline of the article}
The present article is structured as follows: in Section
\ref{Section:RNBasics} we present a discussion of the features of the
Reissner-Nordstr\"om spacetime that will be used in our present
analysis. Section \ref{Section:ConformalCurves} provides a
discussion of the properties of the class of conformal curves that
will be used to study the conformal properties of electrovacuum
spacetimes. Section \ref{Section:ExplicitExpressions} particularises the expressions of Section
\ref{Section:ConformalCurves} to the case of the
Reissner-Nordstr\"om, and establishes general properties of the
congruence under consideration. Section \ref{Section:AnalysisConformalCurves} contains the main
results concerning the behaviour of the individual curves of the
congruence. Section \ref{Section:ConformalDeviationEquations} analyses the behaviour of the deviation
equation of the congruence of conformal curves and provides the proof
of the fact that for the congruence under consideration the curves do
not intersect in the domain of outer communication of the black hole
spacetime. Finally, Section \ref{Section:Conclusions} provides some concluding remark to
our analysis. The article contains an appendix in which the behaviour
of conformal geodesics in the Schwarzschild spacetime is analysed in
a way alternative to that of reference \cite{Fri03c}.

\subsection*{Notations and conventions}
In what follows $\mu,\,\nu,\ldots$ will denote spacetime tensorial
indices. The indices $\alpha,\,\beta,\ldots$ are
spatial tensorial indices. The signature
convention for the spacetime metrics is $(+,-,-,-)$. Thus, the
induced metrics on spacelike hypersurfaces are negative definite. The
Latin indices $i,\,j,\dots$ denote spacetime frame indices taking the
values $0,\ldots,3$, while
$a,\,b,\ldots$ correspond to spatial frame ones ranging over $1,\,2,\,3$.

\medskip
An index-free notation will be often used. Given a 1-form ${\bm
\omega}$ and a vector ${\bm v}$, we denote their contraction by
$\langle {\bm \omega},{\bm v}\rangle$. Furthermore, ${\bm
\omega}^\sharp$ and ${\bm v}^\flat$ denote, respectively, the
contravariant version of ${\bm \omega}$ and the covariant version of
${\bm v}$. The metric with respect to which the operation of raising/lowering
indices will be clear by the context. 

\medskip
 In order to ease the
presentation some of the notation used in \cite{Fri03c} for the
various types of coordinates has been modified.

\section{The Reissner-Nordstr\"om spacetimes}
\label{Section:RNBasics}

We begin by recalling that the Einstein-Maxwell field
equations with vanishing Cosmological constant are given by 
\begin{subequations}
\begin{eqnarray}
&& \tilde{R}_{\mu\nu} -\tfrac{1}{2}\tilde{R}\, \tilde{g}_{\mu\nu}  =
\tilde{F}_{\mu\lambda} \updn{\tilde{F}}{\lambda}{\nu}
-\tfrac{1}{4}\tilde{g}_{\mu\nu}
\tilde{F}_{\lambda\rho}\tilde{F}^{\lambda\rho}, \label{EMFE1}\\
&& \tilde{\nabla}^\mu \tilde{F}_{\mu\nu}=0, \label{EMFE2}\\
&& \tilde{\nabla}_{[\mu} \tilde{F}_{\nu\lambda]}=0, \label{EMFE3}
\end{eqnarray}
\end{subequations}
where $\tilde{R}_{\mu\nu}$ denotes the Ricci tensor of the Lorentzian
metric $\tilde{g}_{\mu\nu}$, and $\tilde{F}_{\mu\nu}$ is the Faraday
tensor. In view of \emph{Birkhoff's theorem for electrovacuum spacetimes}
---see e.g. \cite{SKMHH} page 232--- the Reissner-Norstr\"om spacetime is
the only spherically symmetric solution to equations
\eqref{EMFE1}-\eqref{EMFE3}.

\subsection{Basic expressions and coordinates}
In what follows, we briefly discuss various coordinates
representations of the Reissner-Nordstr\"om spacetime that will be
used in the sequel. Further details can be found in
e.g. \cite{GriPod09,HawEll73}.

\subsubsection{Standard coordinates}
In standard \emph{spherical coordinates} $(t,r,\theta,\varphi)$, the
line element and the Faraday tensor of the Reissner-Nordstr\"om
spacetime is given by
\begin{subequations}
\begin{eqnarray}
&& \hspace{-10mm}{\tilde{\bm g}} = \left(
  1-\frac{2m}{r}+\frac{q^2}{r^2} \right) \mbox{\bf
  d}t\otimes \mbox{\bf d}t -\left( 1-\frac{2m}{r}+\frac{q^2}{r^2}
\right)^{-1} \mbox{\bf d} r\otimes \mbox{\bf d} r
- r^2 {\bm\sigma}^2 \label{StandardRN}\\
&& \hspace{-10mm} \tilde{\bm F} =
\frac{q}{2r^2} \mbox{\bf d}t \wedge \mbox{\bf d}r.
\end{eqnarray}
\end{subequations}
where 
\[
{\bm \sigma} \equiv \left(\mbox{\bf d}
  \theta \otimes \mbox{\bf d} \theta + \sin^2 \theta \mbox{\bf
    d}\varphi \otimes \mbox{\bf d} \varphi \right)
\]
is the standard metric of $\mathbb{S}^2$. All throughout it is assumed that 
\[
m>0, \quad m^2\geq q^2,
\]
so that the solution describes a black hole. If $q=0$, the line
element \eqref{StandardRN} reduces to the corresponding one of the
Schwarzschild spacetime. As this case was analysed in detail in
\cite{Fri03c}, we assume $q\neq 0$ unless explicitly stated. The
\emph{extremal case}, $q^2=m^2$, is of particular interest. In that
case the line element reduces to
\[
 {\tilde{\bm g}} = \left(
  1-\frac{m}{r} \right)^2 \mbox{\bf
  d}t\otimes \mbox{\bf d}t -\left( 1-\frac{m}{r}
\right)^{-2} \!\!\mbox{\bf d} r\otimes \mbox{\bf d} r -
r^2{\bm \sigma}.
\]

\medskip
In view of our subsequent discussion we define \footnote{The function
  $D(r)$ corresponds to the function $F(\bar{r})$ of reference
  \cite{Fri03c}. A different notation has been introduced to avoid
  confusion with the Faraday tensor $\tilde{\bm F}$.}
\begin{equation}
D(r)\equiv \left(1-\frac{2m}{r} +\frac{q^2}{r^2} \right) =
\frac{1}{r^2}(r-r_+)(r-r_-), 
\label{DefinitionD}
\end{equation}
where
\[
r_\pm \equiv m \pm \sqrt{m^2-q^2}.
\]
As it is well known, the locus of points in the spacetime for which
$r=r_+$ and $r=r_-$ correspond, respectively,
to the \emph{event horizon} and the \emph{Cauchy horizon}. Notice that
$D(r_\pm)=0$. In the
extremal case we have that $r_\pm =m$ so that the event and
Cauchy horizons coincide.  Notice that in the non-extremal case $0<r_-<r_+$,
and that
\begin{eqnarray*}
D(r) >0 && \mbox{if }   \quad r_-< r_+ <
r, \\
D(r)<0 && \mbox{if }    \quad  r_- < r < r_+, \\
D(r)>0 && \mbox{if }  \quad 0<r < r_- < r_+.
\end{eqnarray*}
However, in the extremal case 
\[
D(r)>0 \quad \mbox{ for } \quad 0<r<m \quad \mbox{ and } \quad m<r.
\]

\subsubsection{Isotropic coordinates}
An isotropic coordinate $\varrho$ can be introduced in the line element
\eqref{StandardRN} via the requirement
\begin{equation}
\frac{\mbox{d}\varrho}{\mbox{d}r} = \frac{\varrho}{r \sqrt{D}}.
\label{IsotropicCoordinatesDE}
\end{equation}
The latter condition implies, for $r>r_+ , \,
r_-$, the coordinate transformation 
\begin{equation}
\varrho= \frac{1}{2}\left(r-m + \sqrt{r^2-2mr+q^2} \right), \qquad r =\frac{1}{4\varrho}(2\varrho+m+q)(2\varrho+m-q),
\label{IsotropicCoordinates}
\end{equation}
so as to obtain the line element
\[
\tilde{\bm g}= \frac{\displaystyle\left(
    1+\frac{q^2-m^2}{4\varrho^2}
  \right)^2}{\displaystyle\left(1+\frac{m+q}{2\varrho}\right)^2\left(1+\frac{m-q}{2\varrho}\right)^2}\mbox{\bf
  d}t \otimes \mbox{\bf d}t
-\left(1+\frac{m+q}{2\varrho}\right)^2 \left(
  1+\frac{m-q}{2\varrho}\right)^2\left( \mbox{\bf d} \varrho \otimes \mbox{\bf d} \varrho + \varrho^2
  {\bm \sigma} \right).
\]
In the extremal case the isotropic coordinate
transformation reduces to
\[
\varrho=r-m, \qquad r=\varrho+m,
\]
and the corresponding line element is given by
\[
\tilde{\bm g} = \left(1+\frac{m}{\varrho}\right)^{-2}\mbox{\bf d}t \otimes
\mbox{\bf d}t -\left( 1+\frac{m}{\varrho} \right)^2\left( \mbox{d}\varrho\otimes
  \mbox{d}\varrho + \varrho^2 {\bm \sigma}\right).
\]
For future reference it is noticed that
\[
D(\varrho) \equiv D(r(\varrho)) = \frac{\displaystyle\left(
    1+\frac{q^2-m^2}{4\varrho^2} \right)^2}{\displaystyle\left(1+\frac{m+q}{2\varrho}\right)^2\left(1+\frac{m-q}{2\varrho}\right)^2},
\]
and that
\[
\qquad \varrho_\pm\equiv \varrho(r_\pm) = \pm\tfrac{1}{2}\sqrt{m^2-q^2}.
\]
In particular, in the extremal case one has
\[
D(\varrho) = \left(1+\frac{m}{\varrho}\right)^{-2}.
\]

\subsubsection{Null coordinates}
\label{Section:NullCoordinates}
Eddington-Finkelstein-like null coordinates can be introduced in the
non-extremal case via
\begin{subequations}
\begin{eqnarray}
&& u= t -\left(r+ \frac{r_+^2}{r_+-r_-}\ln|r-r_+| -
\frac{r_-^2}{r_+-r_-}\ln|r-r_-|\right), \label{NullCoordinates1} \\
&& v = t +\left(r+ \frac{r_+^2}{r_+-r_-}\ln|r-r_+| -
\frac{r_-^2}{r_+-r_-}\ln|r-r_-|\right), \label{NullCoordinates2}
\end{eqnarray}
\end{subequations}
so that one obtains the line elements
\begin{equation}
\tilde{\bm g} =  D(r) \mbox{\bf
  d}u\otimes \mbox{\bf d}u + 2 \mbox{ \bf d} u\otimes \mbox{\bf
  d}r - r^2 {\bm \sigma}, \quad \tilde{\bm g} = D(r) \mbox{\bf
  d}v\otimes \mbox{\bf d}v - 2 \mbox{ \bf d} v\otimes \mbox{\bf
  d}r - r^2 {\bm \sigma}.
\label{EFRN}
\end{equation}
In the extremal case the corresponding change of coordinates leading
to line elements of the form given in \eqref{EFRN} is given by
\begin{equation}
u = t -\left( r -\frac{m^2}{r-m} +2m \ln|r-m| \right), \quad v=t+\left(
  r -\frac{m^2}{r-m} +2m \ln|r-m|  \right).
\label{NullCoordinates3}
\end{equation}

\subsection{Conformal diagrams of the Reissner-Nordstr\"om spacetime}
The conformal Reissner-Nordstr\"om spacetime in both the non-extremal and the
extremal case are well known ---see for example
\cite{GriPod09,HawEll73,Car73a} for details on their
construction. These diagrams are included in Figure
\ref{ConformalDiagramRN} for quick reference. 

\begin{figure}[t]
\centerline{\includegraphics[width=\textwidth]{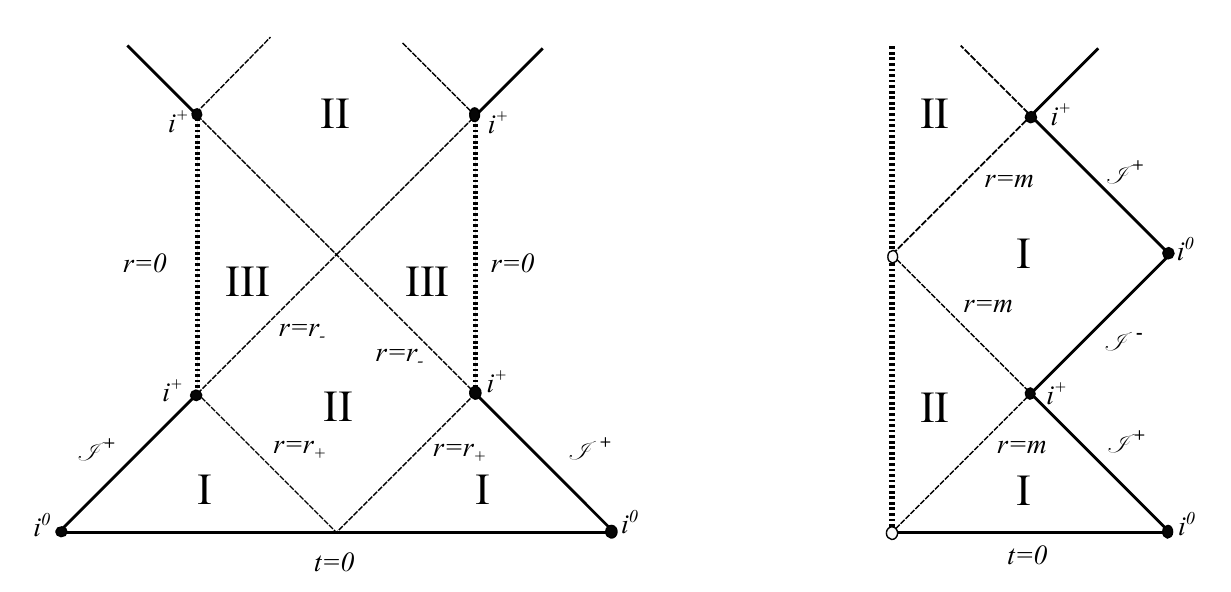}}
\caption{Conformal diagrams of the Reissner-Nordstr\"om
  spacetime: the non-extremal case (left) and the extremal case (right).}
\label{ConformalDiagramRN}
\end{figure}

\subsection{Time symmetric initial data}
\label{Section:PhysicalData}
In the sequel, it will be shown how to cover portions of the
Reissner-Nordstr\"om spacetime by means of a congruence of curves whose initial
data is prescribed on the time symmetric slice $\tilde{\mathcal{S}}=\{t=\mbox{constant}\}$
of the spacetime. The hypersurface $\tilde{\mathcal{S}}$ has the
  topology of $\mathbb{R}^3\setminus\{0\}$. The initial 3-metric
  expressed in isotropic coordinates takes the form
\begin{equation}
\tilde{\bm h}= -\phi^2 \chi^2 \left( \mbox{\bf d}\varrho \otimes
  \mbox{\bf d}\varrho +\varrho^2 {\bm \sigma} \right),
\label{InitialMetricRN}
\end{equation}
with
\[
\phi \equiv \left(1+\frac{m+q}{2\varrho}\right), \quad \chi\equiv
\left( 1+\frac{m-q}{2\varrho}\right), \quad \varrho\in[0,\infty).
\]
If $m^2>q^2$ then the Riemannian 3-dimensional manifold given by
$(\tilde{\mathcal{S}},\tilde{\bm h})$ has two asymptotically Euclidean
ends joined by a throat at $\varrho=\varrho_+$. In the
extremal case one has either $\phi=1$ or $\chi=1$. Accordingly,
$(\tilde{S},\tilde{\bm h})$
has one asymptotically Euclidean end and one
so-called \emph{trumpet-like} end. 

\section{A class of conformal curves}
\label{Section:ConformalCurves}

In what follows, we present a brief discussion of the properties of a
class of \emph{conformal curves} introduced in \cite{LueVal12}. These
curves are a generalisation of the conformal geodesics allowing to
recover, for non-vacuum spacetimes, some of the properties satisfied by
conformal geodesics in a vacuum spacetime.  Our discussion of the
properties of conformal curves is inspired by and follows closely the
discussion of conformal geodesics in vacuum spacetimes given in \cite{Fri95,Fri03c}.

\medskip
Although the notion of conformal curves is applicable to spacetimes
with an arbitrary matter models, in what follows, for concreteness it
is assumed that $(\tilde{\mathcal{M}},\tilde{\bm g})$ denotes a
spacetime satisfying the Einstein-Maxwell field equations \eqref{EMFE1}-\eqref{EMFE3}. Let
$(\mathcal{M},{\bm g})$ denote a conformal extension of
$(\tilde{\mathcal{M}},\tilde{\bm g})$. Hence, there exists a scalar
$\Theta$ such that the metrics $\tilde{\bm g}$ and
${\bm g}$ are related via
\begin{equation}
{\bm g} = \Theta^2 \tilde{\bm g}.
\label{ConformalTransformation:metric}
\end{equation}

\subsection{Basic definitions}
\label{Section:ConformalCurvesBasics}

Given an interval $I\subseteq\mathbb{R}$, let ${\bm x}(\tau)$, $\tau
\in I$ denote a curve in $(\tilde{\mathcal{M}},\tilde{\bm g})$ and let ${\bm
b}(\tau)$ denote a 1-form along ${\bm x}(\tau)$. Furthermore, let
$\dot{\bm x}\equiv \mbox{d}{\bm x}/\mbox{d}\tau$ denote the tangent
vector field of the curve ${\bm x}(\tau)$. The pairing between the
vector $\dot{\bm x}$ and the 1-form $\bm b$ is denoted by
$\langle{\bm b},\dot{\bm x }\rangle$. In \cite{LueVal12} the following
equations for the pair $({\bm x}(\tau),{\bm b}(\tau))$ have been
introduced:
\begin{subequations}
\begin{eqnarray}
&& \tilde{\nabla}_{\dot{\bm x}} \dot{\bm x} = -2 \langle {\bm b},
\dot{\bm x} \rangle
\dot{\bm x} + \tilde{\bm g}(\dot{\bm x},\dot{\bm x}) {\bm b}^\sharp, \label{ConformalCurve1} \\
&& \tilde{\nabla}_{\dot{\bm x}} {\bm b} = \langle {\bm b}, \dot{\bm x}
\rangle {\bm b} - \tfrac{1}{2} \tilde{\bm g}^\sharp ({\bm b},{\bm b})
\dot{\bm x}^\flat + \tilde{\bm H}(\dot{\bm x}, \cdot), \label{ConformalCurve2}
\end{eqnarray}
\end{subequations}
where $\tilde{\nabla}_{\dot{x}}$ denotes the directional derivative of
the Levi-Civita connection of the metric $\tilde{\bm g}$, while 
$\tilde{\bm H}$ denotes a rank 2 covariant tensor which upon the
conformal transformation \eqref{ConformalTransformation:metric}
transforms as:
\[
\tilde{H}_{\mu\nu}-H_{\mu\nu} = \nabla_\mu \Upsilon_\nu
+\Upsilon_\mu\Upsilon_\nu -\tfrac{1}{2}
g^{\lambda\rho}\Upsilon_\lambda \Upsilon_\rho \, g_{\mu\nu}, \qquad
\Upsilon_\mu \equiv \Theta^{-1}\nabla_\mu \Theta.
\]
This transformation law is formally identical to that of the
\emph{Schouten tensor} 
\[
\tilde{L}_{\mu\nu} \equiv\tfrac{1}{2}(
\tilde{R}_{\mu\nu} - \tfrac{1}{6} \tilde{R}\, \tilde{g}_{\mu\nu}).
\] 
The equations \eqref{ConformalCurve1}-\eqref{ConformalCurve2} will be
referred to as the \emph{conformal curve equations}. In a slight abuse
of notation, the pair $({\bm x}(\tau),{\bm b}(\tau))$ will be called a
\emph{conformal curve}. Following \cite{LueVal12} we set 
\begin{equation}
\tilde{\bm H}=0,
\label{ChoiceH}
\end{equation}
has been adopted so that equation \eqref{ConformalCurve2} reduces to
\begin{equation}
\tilde{\nabla}_{\dot{\bm x}} {\bm b} = \langle {\bm b}, \dot{\bm x}
\rangle {\bm b} - \tfrac{1}{2} \tilde{\bm g}^\sharp ({\bm b},{\bm b})
\dot{\bm x}^\flat. \label{ConformalCurve2ALT}
\end{equation}
\emph{In what follows, the choice \eqref{ChoiceH} will be assumed.}

\medskip
The choice given by equation \eqref{ChoiceH}  leads to
an explicit expression for the conformal factor $\Theta$ in terms of
the parameter $\tau$. Indeed, by requiring ${\bm x}(\tau)$ to be
timelike and imposing the normalisation condition
\begin{equation}
{\bm g}(\dot{\bm x},\dot{\bm x}) =1,
\label{Normalisation:Unphysical}
\end{equation}
the conformal curve equations imply 
\[
\dot{\Theta} = \langle {\bm b}, \dot{\bm x}\rangle \Theta, \qquad
\ddot{\Theta} = \tfrac{1}{2} \tilde{\bm g}^\sharp({\bm b},{\bm b})
\Theta^{-1}, \qquad \dddot{\Theta}=0,
\]
where $\dot{\Theta} \equiv \tilde{\nabla}_{\dot{x}} \Theta$,
etc. Integrating the last of these equations one finds 
\begin{equation}
\Theta = \Theta_* + \dot{\Theta}_* (\tau-\tau_*) +
\tfrac{1}{2}\ddot{\Theta}_*(\tau-\tau_*)^2, \label{ConformalFactor}
\end{equation}
where $\Theta_*$, $\dot{\Theta}_*$ and $\ddot{\Theta}_*$ are
prescribed at a fiduciary value $\tau_*$ of the parameter
$\tau$. The coefficients $\dot{\Theta}_*$ and $\ddot{\Theta}_*$
satisfy the constraints
\begin{equation}
\dot{\Theta}_* = \langle {\bm b}_*, \dot{\bm x}_* \rangle \Theta_*, \qquad 2
\Theta_* \ddot{\Theta}_* = \tilde{\bm g}({\bm b}_*,{\bm b}_*),
\label{Constraints}
\end{equation}
where ${\bm b}_*$ and $\dot{\bm x}_*$ denote, respectively, the value
of ${\bm b}$ and $\dot{\bm x}$ at $\tau=\tau_*$. 

\medskip
Finally, let ${\bm e}_i$, $i=0,\ldots,3$, denote a frame basis along
${\bm x}(\tau)$.  The frame will be said to be \emph{Weyl
propagated} along the conformal curve $({\bm x}(\tau),{\bm b}(\tau))$ if
it satisfies the equation
\begin{equation}
\tilde{\nabla}_{\dot{x}}  {\bm e}_i = -\langle {\bm b}, {\bm e}_i
\rangle\dot{\bm x} -\langle {\bm b},{\bm e}_i \rangle {\bm e}_i +
\tilde{g}({\bm e}_i,\dot{\bm x}){\bm b}^\sharp. 
\end{equation}
It can be readily seen that a Weyl propagated frame which is
${\bm g}$-orthonormal at, say, $\tau=\tau_*$ remains ${\bm g}$-orthonormal all through
${\bm x}(\tau)$ ---that
is, ${\bm g}({\bm e}_i,{\bm e}_j)=\eta_{ij}$. Following the discussion
of \cite{Fri03c,LueVal12},
If, consistently with equation \eqref{Normalisation:Unphysical}, one sets ${\bm
  e}_0=\dot{\bm x}$ then it can be shown that
\[
b_0 =\Theta^{-1} \dot{\Theta}, \qquad b_a = \langle \Theta^{-1} {\bm
  b}, {\bm e}_a \rangle_*, \qquad a=1,\,2,\,3.
\]
As a consequence, the components of the 1-form $\bm b$ with respect to
the frame ${\bm e}_i$ can be expressed in terms of the value of
various fields at $\tau=\tau_*$. Hence, the 1-form ${\bm b}$, like $\Theta$, is
known \emph{a priori}. For full details of the computations involved
see e.g. \cite{Fri03c,LueVal12}.

\medskip
\noindent
\textbf{Remark 1.}
Notice that by virtue of the normalisation condition
\eqref{Normalisation:Unphysical}, $\tau$ is the ${\bm g}$-proper time
of the conformal curve. We shall often refer to $\tau$ as the
\emph{unphysical proper time}. 

\medskip
\noindent
\textbf{Remark 2.} With the choice $\tilde{\bm H}=\tilde{\bm L}$, the
equations \eqref{ConformalCurve1} and \eqref{ConformalCurve2} yield the
so-called \emph{conformal geodesic equations} ---see
e.g. \cite{FriSch87,Fri95,Fri03c}. Notice, that for a vacuum spacetime
$\tilde{\bm L}=0$ and the conformal geodesic equations are formally
identical to equations
\eqref{ConformalCurve1}-\eqref{ConformalCurve2ALT}. In this case, it was shown in
Lemma 3.1 of \cite{Fri95} that a quadratic expression for the
conformal factor identical to \eqref{ConformalFactor} can be
obtained. Finally, it is observed that in the case of an
electrovacuum spacetime one has that $\tilde{\bm L}\neq 0$ and the
argument leading to Lemma 3.1 in \cite{Fri95} no longer holds. The
desire of retaining the expression \eqref{ConformalFactor} is what led
in \cite{LueVal12} to the notion of conformal curves.

\subsection{The $\tilde{\bm g}$-adapted equations}
\label{Section:PhysicalMetricAdaptedEquations}

As already mentioned, as a consequence of the normalisation condition
\eqref{Normalisation:Unphysical}, the parameter $\tau$ is the unphysical proper
time of the curve ${\bm x}(\tau)$. In some computations it is more
convenient to consider a parametrisation in terms of \emph{physical}
proper time $\bar{\tau}$. The parameter transformation
\begin{equation}
\bar{\tau} = \bar{\tau}_* + \int_{\tau_*}^{\tau} \frac{\mbox{d} s}{\Theta(s)},
\label{PhysicalTime}
\end{equation}
with inverse $\tau=\tau(\bar{\tau})$. In what follows, we will write
$\bar{\bm x}\equiv {\bm x}(\tau(\bar{\tau}))$. It can then
be verified
that
\[
\bar{\bm x}'
\equiv \partial_{\bar{\tau}} \bar{\bm x} = \Theta \dot{\bm x},
\]
and that $\tilde{\bm g}(\bar{\bm x}',\bar{\bm x}')=1$. Hence, 
$\bar{\tau}$ is, indeed, the $\tilde{\bm g}$-proper time of the curve
$\bar{\bm x}$.  

\medskip
Now, in order to write the equation for the curve $\bar{\bm
  x}(\bar{\tau})$ in a convenient way, we consider the split
\[
{\bm b}=\hat{\bm b} + \varpi \dot{\bm x}^\flat,
\]
where the 1-form $\hat{\bm b}$ satisfies
\[
\langle \hat{\bm b}, \dot{\bm x}\rangle =0, \qquad \varpi =
\frac{\langle {\bm b},
  \dot{\bm x}\rangle}{\tilde{\bm g}(\dot{\bm x},\dot{\bm x})}, \qquad {\bm
  g}^\sharp({\bm b},{\bm b})
=\langle {\bm b},\dot{\bm x}\rangle^2 + {\bm g}^\sharp(\hat{\bm
  b},\hat{\bm b}),
\] 
and the indices of the vectors and forms have been moved using the
metric $\tilde{\bm g}$. In terms of these
objects the $\tilde{\bm g}$-\emph{adapted equations for the conformal curves} are
given by
\begin{subequations}
\begin{eqnarray}
&& \tilde{\nabla}_{\bar{\bm x}'} \bar{\bm x}' = \hat{\bm b}^\sharp, \label{PhysicalMetricAdaptedCC1}\\
&& \tilde{\nabla}_{\bar{\bm x}'} \hat{\bm b} = \beta^2 \bar{\bm
  x}'{}^\flat, \label{PhysicalMetricAdaptedCC2}
\end{eqnarray}
\end{subequations}
where 
\begin{equation}
\beta^2 \equiv -\tilde{\bm g}^\sharp(\hat{\bm b},\hat{\bm b}) = \delta^{ab} d_a d_b = \mbox{constant}
\label{Definition:beta}
\end{equation}
is, by virtue of the discussion of Section
\ref{Section:ConformalCurvesBasics}, a constant along the conformal
curve. Finally, it is worth noticing that as a consequence of equation
\eqref{PhysicalMetricAdaptedCC1}, $\hat{\bm b}^\sharp$ can be
interpreted as the \emph{physical acceleration} of the conformal
curve.

\subsection{The deviation equations}
\label{Section:ConformalDeviationEquations}

A crucial part of the analysis of the present article will be concerned
with the question of whether a
congruence of conformal curves develops conjugate points or not. To this end,
let $({\bm x}(\bar{\tau},\lambda),{\bm b}(\bar{\tau},\lambda))$ denote a family of
conformal curves depending smoothly on a parameter $\lambda$. Following \cite{Fri03c}, let
\[
\bar{\bm X} \equiv \bar{\bm x}', \qquad \bar{\bm Z}
\equiv \partial_\lambda \bar{\bm x},
\qquad \hat{\bm B} \equiv \tilde{\nabla}_{\bar{\bm Z}} \hat{\bm b}.
\]
One then has that
\begin{subequations}
\begin{eqnarray}
&& \tilde{\nabla}_{\bar{\bm X}} \tilde{\nabla}_{\bar{\bm X}} \bar{\bm Z} =
\tilde{\bm R}(\bar{\bm X},\bar{\bm Z})\bar{\bm X} +
\hat{\bm B}^\sharp, \label{PhysicalMetricAdaptedDevEqns1} \\
&& \tilde{\nabla}_{\bar{\bm X}} \hat{\bm B} = -\hat{\bm b} \cdot \tilde{\bm
  R}(\bar{\bm X},\bar{\bm Z})+
\tilde{\nabla}_{\bar{\bm Z}} \big(\tilde{\bm g}^\sharp(\hat{\bm b},\hat{\bm
    b})\big) \bar{\bm X}^\flat +
\tilde{\bm g}^\sharp(\hat{\bm b},\hat{\bm b}) \tilde{\nabla}_{\bar{\bm
    X}} \bar{\bm Z}^\flat. \label{PhysicalMetricAdaptedDevEqns2}
\end{eqnarray}
\end{subequations}
A computation then shows that
\begin{equation}
\tilde{\nabla}_{\bar{\bm X}} \tilde{\nabla}_{\bar{\bm Z}}
\big(\tilde{\bm g}^\sharp(\hat{\bm b},\hat{\bm b})\big) =
\tilde{\nabla}_{\bar{\bm Z}} \tilde{\nabla}_{\bar{\bm X}}
\big(\tilde{\bm g}^\sharp(\hat{\bm b},\hat{\bm b})\big)=0,
\label{ConstantsAlongCurve}
\end{equation}
as a consequence of equation \eqref{Definition:beta}. Hence, one concludes
that the coefficients $\tilde{\nabla}_{\bar{\bm Z}} \big(\tilde{\bm g}^\sharp(\hat{\bm b},\hat{\bm
    b})\big)$ and $\tilde{\bm g}^\sharp(\hat{\bm b},\hat{\bm b})$ are
  constant along a conformal curve. For
simplicity one can evaluate them at $\tau=\tau_*$.

\subsection{Formulae in warped product spaces}

The Reissner-Nordstr\"om spacetime in the standard coordinates
$(t,r,\theta,\varphi)$ of the line element of equation \eqref{StandardRN}
is in the form of a warped product. This structure can be exploited to
simplify the analysis of the $\tilde{\bm g}$-adapted conformal curve
equations
\eqref{PhysicalMetricAdaptedCC1}-\eqref{PhysicalMetricAdaptedCC2}. In
this section, we adapt the discussion of \cite{Fri03c} to the context
of the conformal curves. 

\medskip
In what follows, we will consider spacetimes whose metric can be
written as a warped product of the form
\begin{equation}
\tilde{\bm g}  = l_{AB}
\mbox{\bf d}x^A \otimes \mbox{\bf d} x^B + f^2 k_{cd} \mbox{\bf
  d}x^c\otimes  \mbox{\bf d}x^d,
\label{WarpedProductMetric}
\end{equation}
with 
\[
l_{AB} = l_{AB} (x^C), \quad k_{ab} = k_{ab}(x^c), \quad f=f(x^A)>0,
\]
and $A,\, B, \,C=0,\,1$ and $a,\, b,\, c=2,\,3$. In addition, it is
assumed that  the 2-dimensional metric given by the line element ${\bm
  l}\equiv l_{AB} \mbox{\bf d}x^A \otimes \mbox{\bf d} x^B$ is
Lorentzian, while the one given by ${\bm k} =k_{cd} \mbox{\bf
  d}x^c\otimes  \mbox{\bf d}x^d$ is a negative-definite Riemannian
one. \emph{In view of this
structure it is natural to consider solutions to the conformal curve
equations satisfying $\dot{x}^a=0$ and ${b}_c=0$.} In the context of
the spherically symmetric Reissner-Nordstr\"om spacetime this Ansatz
leads to solutions to the conformal curves whose angular coordinates
are constant. One only has to consider evolution
equations for the coordinates $(r,t)$. A direct
computation shows that for this type of conformal curves the
$\tilde{\bm g}$-adapted equations for the conformal curves imply
\begin{subequations}
\begin{eqnarray}
&& \not{\!\!D}_{\bar{\bm x}'} \bar{\bm x}' = \hat{\bm b}^\sharp, \label{WarpedProductConformalCurveEquations1}\\
&& \hat{\bm b} = \pm \beta {\bm\epsilon}_{\bm l}(\bar{\bm
  x}',\,\cdot\,), \label{WarpedProductConformalCurveEquations2}
\end{eqnarray}
\end{subequations}
with 
\[
{\bm \epsilon}_{\bm l} \equiv \sqrt{|\Delta|} \mbox{\bf d}x^0 \wedge
\mbox{\bf d}x^1, \quad
\Delta\equiv \det l_{AB},
\]
and where $\not{\!\!D}$ denotes the Levi-Civita covariant derivative
of ${\bm l}$ and $\beta^2$ given by equation
\eqref{Definition:beta}. The sign in equation
\eqref{WarpedProductConformalCurveEquations2} is determined
consistently with the initial conditions. 
A further computation shows that under the present Ansatz, the
$\tilde{\bm g}$-adapted deviation equations
\eqref{PhysicalMetricAdaptedDevEqns1}-\eqref{PhysicalMetricAdaptedDevEqns2}
are equivalent to each other and, in turn, to the equation
\begin{equation}
\not{\!\!D}_{\bar{\bm X}} \not{\!\!D}_{\bar{\bm X}} \bar{\bm Z} = \tfrac{1}{2} R[{\bm l}]\, {\bm
  \epsilon}_{\bm l}(\bar{\bm X},\bar{\bm Z}) {\bm \epsilon}_{\bm l}(\bar{\bm X},
  \,\cdot\,) ^\sharp \pm \left( \not{\!\!D}_{\bar{\bm Z}} \beta\, {\bm \epsilon}_{\bm l}
  (\bar{\bm X},\, \cdot\,)^\sharp + \beta {\bm \epsilon}_{\bm l}
  (\not{\!\!D}_{\bar{\bm X}} \bar{\bm
  Z},\,\cdot \,)^\sharp\right),
\label{WarpProductDeviationEquation}
\end{equation}
where $R[{\bm l}]$ denotes the Ricci scalar of ${\bm l}$.  The sign in
the last equation is chosen consistently with that of equation
\eqref{WarpedProductConformalCurveEquations2}. 

\medskip
For conformal curves satisfying $\dot{x}^a=0$ and ${b}_c=0$, the
issue of whether the deviation vector field $\bar{\bm Z}$
degenerates can be rephrased in terms of the question of the vanishing
of the scalar 
\begin{equation}
\omega \equiv {\bm \epsilon}_{\bm l}(\bar{\bm X},\bar{\bm Z}).
\label{Definition:omega}
\end{equation}
Notice that as long as $\omega\neq 0$, $\bar{\bm X}$ and $\bar{\bm Z}$
are linearly independent. A computation using
\eqref{WarpProductDeviationEquation} yields
\begin{equation}
\not{\!\!D}_{\bar{\bm X}} \not{\!\!D}_{\bar{\bm X}}\omega
  = \left( \beta^2 + \tfrac{1}{2}R[{\bm l}]\right)\omega + \not{\!\!D}_{\bar{\bm Z}} \beta.
\label{ReducedWarpProductDeviationEquation}
\end{equation}

\section{Basic expressions for the congruence of conformal curves in
  the Reissner-Nordstr\"om spacetime}
\label{Section:ExplicitExpressions}

In the present section we particularise the discussion made in Section
\ref{Section:ConformalCurves} to a specific class of
conformal curves in the Reissner-Nordstr\"om spacetime.

\subsection{Initial data for the congruence}
\label{Section:InitialDataCongruence}

It follows from the discussion in Section
\ref{Section:ConformalCurvesBasics} that basic pieces of information
to be prescribed in order to construct a congruence of conformal
curves are the initial value of the conformal factor, $\Theta_*$, and
the initial value of the 1-form ${\bm b}_*$ at some initial
hypersurface. Following the discussion of Section
\ref{Section:PhysicalData} we consider the time symmetric slice of the
Reissner-Nordstr\"om spacetime. By analogy to the discussion
in \cite{Fri03c}, and taking into account the 3-metric of equation \eqref{InitialMetricRN}, we choose
\begin{eqnarray*}
&& \Theta_* = \phi^{-1} \chi^{-1} = \frac{1}{r_*^2}= \frac{\varrho^4_*}{\left(\varrho_*+
    \displaystyle\frac{m+q}{2}\right)^2\left(\varrho_*+\displaystyle\frac{m-q}{2}\right)^2},
\\
&& {\bm b}_* =\hat{\bm b}_* = \Theta^{-1}_* \mbox{d}\Theta_* =
-\frac{2}{r_*} \mbox{\bf d}r_*= -
\frac{2\left(\varrho_*^2 -\displaystyle\frac{1}{4}(m^2-q^2)\right)}{\varrho_*\left( \varrho_*+
    \displaystyle\frac{m+q}{2}\right)\left( \varrho_*+\displaystyle\frac{m-q}{2}\right)}\mbox{d}\varrho_*,
\end{eqnarray*}
where it is assumed that $r_*> r_+$. The symbols $r_*$ and $\varrho_*$
are used to denote, respectively, the radial coordinates $r$ and
$\rho$ on the time symmetric slice. For simplicity we also set
$\tau_*=0$. In addition, one has that
\[
\beta^2 \equiv -\tilde{\bm g}^\sharp(\hat{\bm b}_*,\hat{\bm b}_*) =
\frac{4}{r^2_*}D_* 
\]
where $D_*\equiv D(r_*)$ and $D(r)$ is given by equation
\eqref{DefinitionD}. If, moreover, one assumes the condition
\[
\langle {\bm b}, \dot{\bm x} \rangle_*=0,
\]
then the expression \eqref{ConformalFactor} for the conformal factor,
together with the constraints \eqref{Constraints} imply
\begin{equation}
\Theta = D_* \left( \left(\frac{2\Theta_*}{\beta} \right)^2-\tau^2 \right),
\label{ConformalFactorReduced}
\end{equation}
where the coefficients $D_*$, $\Theta_*$ and $\beta$ in this expression
are taken to be constant along a given conformal curve. Moreover, the
conformal curves of the congruence generated by the above conditions
can be parametrised by the value of their radial coordinate on the
initial hypersurface, $r_*$ (or $\varrho_*$). Using expression
\eqref{ConformalFactorReduced} in formula \eqref{PhysicalTime} one finds that the parameters $\tau$ and
$\bar{\tau}$ are related to each other by
\begin{equation}
\bar{\tau}= \frac{1}{\beta}\ln\left(\frac{2\Theta_* +\beta \tau}{2\Theta_*-\beta \tau}\right).
\label{PhysicalToUnphysicalProperTime}
\end{equation}
The inverse relation giving $\tau$ in terms of $\bar{\tau}$ is given
by
\begin{equation}
\tau = \frac{2\Theta_*}{\beta} \tanh\left( \tfrac{1}{2}\beta \bar{\tau} \right).
\label{UnphysicalToPhysicalProperTime}
\end{equation}

In terms of the physical proper time, $\bar{\tau}$, the expression
\eqref{ConformalFactorReduced} takes the form
\begin{equation}
\Theta = \frac{\Theta_*}{\cosh^2\left(
    \tfrac{1}{2}\beta\bar{\tau}\right)}.
\label{ConformalFactorReducedPhysical}
\end{equation}

\subsection{The conformal curve equations for the standard coordinates}

In order to write the conformal curve equations
\eqref{WarpedProductConformalCurveEquations1}-\eqref{WarpedProductConformalCurveEquations2}
for the Reissner-Nordstr\"om metric, it is noticed that
the metric ${\bm l}$ in the warped product line element
\eqref{WarpedProductMetric} is given by 
\[
{\bm l} = D(r) \mbox{\bf d}t \otimes \mbox{\bf d}t -D^{-1}(r)
\mbox{\bf d}r \otimes \mbox{\bf d}r.  
\]
Equations 
\eqref{WarpedProductConformalCurveEquations1}-\eqref{WarpedProductConformalCurveEquations2} imply
\begin{subequations}
\begin{eqnarray}
&& \bar{t}'' + \frac{\partial_{\bar{r}}D(\bar{r})}{D(\bar{r})}\bar{r}'\bar{t}' = \frac{1}{D(\bar{r})}\beta\, \bar{r}', \label{StandardCGEqn1} \\
&& \bar{r}'' - \frac{\partial_{\bar{r}}D(\bar{r})}{2D(\bar{r})}\bar{r}^{\prime 2} +
\frac{D(\bar{r}) \partial_{\bar{r}}D(\bar{r})}{2}\bar{t}^{\prime 2} = D(\bar{r})\beta\, \bar{t}', \label{StandardCGEqn2}
\end{eqnarray}
\end{subequations}
where consistent with the notation of section
\ref{Section:PhysicalMetricAdaptedEquations} we have set
$\bar{r}\equiv r(\bar{\tau})$, $\bar{t}\equiv t(\bar{\tau})$. 
Initial data for these equations is prescribed by observing the
discussion of Section \ref{Section:InitialDataCongruence}, and by
requiring $\dot{\bm x}$  to be given initially by the unit normal to
$\tilde{\mathcal{S}}$.  It follows that
\begin{equation}
t_*=0, \quad r_*>r_+, \quad \bar{t}'_* = \frac{1}{\sqrt{D_*}}, \quad
\bar{r}'_*=0, \quad (\hat{b}_t)_*=0, \quad (\hat{b}_r)_* =-\frac{2}{r_*},
\label{CGInitialData}
\end{equation}
where $\hat{\bm b}_t \equiv \langle \hat{\bm b},{\bm \partial}_t
\rangle$, $\hat{\bm b}_r \equiv \langle \hat{\bm b},{\bm \partial}_r
\rangle$. Notice that $r_*=\bar{r}_*$, $t_*=\bar{t}_*$. As a
consequence of the symmetry of the hypersurface $\tilde{\mathcal{S}}$
with respect to the bifurcation sphere at $r_*=r_+$, it is only
necessary to consider the case $ r_*>r_+$. The equations
\eqref{StandardCGEqn1}-\eqref{StandardCGEqn2} can be decoupled by
making use of the $\tilde{\bm g}$-normalisation condition
\begin{equation}
D(\bar{r})\, \bar{t}^{\prime 2} - \frac{1}{D(\bar{r})} \bar{r}^{\prime 2} =1.
\label{PhysicalNormalisation}
\end{equation}
Solving the latter for $t'\geq0$ and substituting into \eqref{StandardCGEqn2}, one obtains that
\begin{equation}
\bar{r}'' + \tfrac{1}{2} \partial_{\bar{r}} D(\bar{r}) - \beta \sqrt{D(\bar{r}) +
  \bar{r}^{\prime 2}}=0.
\label{rprimeprime}
\end{equation}
This equation can be integrated once to yield
\[
\sqrt{D(\bar{r}) + \bar{r}^{\prime 2}} -\beta \bar{r} = \gamma,
\]
where $\gamma$ is a constant given in terms of the initial data by
\[
\gamma = -\sqrt{D_*}.
\]
It follows that 
\begin{equation}
\label{ReducedEquation}
\bar{r}' = \pm \sqrt{(\gamma+\beta \bar{r})^2 - D(\bar{r})},
\end{equation}
with the sign depending on the value of $r_*$.

\subsection{Expressions for the conformal curve equations in null
  coordinates}
In order to discuss the behaviour of the conformal curves through null
infinity and the horizon needs to consider the conformal curve
equations
\eqref{WarpedProductConformalCurveEquations1}-\eqref{WarpedProductConformalCurveEquations2}
written in terms of the Eddington-Finkelstein null coordinates of
Section \ref{Section:NullCoordinates}. A computation renders the pairs
of equations
\begin{eqnarray*}
&& \bar{u}'' -\tfrac{1}{2} \partial_{\bar{r}} D(\bar{r}) \bar{u}^{\prime 2} = -\beta \bar{u}', \\
&& \bar{r}'' + \tfrac{1}{2} D(\bar{r}) \partial_{\bar{r}} D(\bar{r}) \, \bar{u}^{\prime 2} + \partial_{\bar{r}} D(\bar{r}) \,
\bar{r}' \bar{u}' = \beta (\bar{r}' + D(\bar{r}) \bar{u}'),
\end{eqnarray*}
and
\begin{eqnarray*}
&& \bar{v}'' + \tfrac{1}{2} \partial_{\bar{r}} D(\bar{r}) \bar{v}^{\prime 2} = \beta \bar{v}', \\
&& \bar{r}'' + \tfrac{1}{2} D(\bar{r}) \partial_{\bar{r}} D(\bar{r})\, \bar{v}^{\prime 2} - \partial_{\bar{r}} D(\bar{r})\, \bar{r}'\bar{v}'
= -\beta (\bar{r}' -D(\bar{r}) \bar{v}'),
\end{eqnarray*}
where $\bar{u}=u(\bar{\tau})$ and $\bar{v}=v(\bar{\tau})$ ---cf. \cite{Fri03c}.
As in the case of the standard Reissner-Nordstr\"om coordinates, the
above expressions can be decoupled using the $\tilde{\bm g}$-normalisation
conditions
\[
D(\bar{r})\, \bar{u}^{\prime 2} + 2 \bar{u}' \bar{r}'=1, \qquad D(\bar{r})\, \bar{v}^{\prime 2} - 2 \bar{v}' \bar{r}' =1, 
\]
which, in turn, yield
\begin{subequations}
\begin{eqnarray}
&& \bar{u}' = \frac{1}{\sqrt{D(\bar{r}) + \bar{r}^{\prime 2}}
  +\bar{r}'}, \label{NullReducedA}\\
&& \bar{v}' = \frac{1}{D(\bar{r})} \left(  \sqrt{D(\bar{r}) + \bar{r}^{\prime 2}} +\bar{r}' \right).
\label{NullReducedB}
\end{eqnarray}
\end{subequations}

In the calculations leading to the above equations, it has been
required that $u'>0$ and $v'>0$ at the initial hypersurface so as to
have future oriented curves. The
initial conditions for these equations are given, for $m^2>q^2$ and $r_*>r_+$, by
\begin{equation}
u_*= -r_*-
\frac{1}{r_+-r_-}\ln\frac{(r_*-r_+)^{r_+^2}}{(r_*-r_-)^{r_-^2}},
\qquad v_* =
r_*+\frac{1}{r_+-r_-}\ln\frac{(r_*-r_+)^{r_+^2}}{(r_*-r_-)^{r_-^2}},
\label{NullCoordinatesDataNonExtremal}
\end{equation}
and in the case $m^2=q^2$ with $r>m$ by
\begin{equation}
u_* = - r_* +\frac{m^2}{r_*-m} -2m \ln(r_*-m), \qquad v_*=
  r_* -\frac{m^2}{r_*-m} +2m \ln(r_*-m) 
\label{NullCoordinatesDataExtremal}
\end{equation}
---see equations \eqref{NullCoordinates1}-\eqref{NullCoordinates2} and
\eqref{NullCoordinates3}. 

\medskip
\noindent
\textbf{Remark.} Once the function $r(\bar{\tau})$ has been
determined, the expressions \eqref{NullReducedA}-\eqref{NullReducedB} allow to obtain
$u(\bar{\tau})$ and $v(\bar{\tau})$ by integration.

\subsection{The polynomials $P(\bar{r})$ and $Q(\bar{r})$}
\label{Section:PolynomialP}

Equation \eqref{ReducedEquation} can be written in the form
\begin{equation}
\bar{r}^{\prime 2} = \frac{1}{\bar{r}^2} P(\bar{r}),
\label{ReducedSqrt}
\end{equation}
where $P(\bar{r})$ is the quartic polynomial 
\[
P(\bar{r}) \equiv \bar{r}^2 \left((\gamma + \beta \bar{r})^2 -D(\bar{r})\right).
\]
Since $\bar{r}'_*=0$, it can be readily verified that $\bar{r}=r_*$ is
a root of $P(\bar{r})$. Hence, we write
\[
P(\bar{r}) = \beta^2 (\bar{r}-r_*)Q(\bar{r}),
\]
with 
\[
Q(\bar{r}) \equiv \bar{r}^3 + \eta \bar{r} + \xi, 
\]
where
\[
\eta \equiv  \frac{1}{\beta^2} (D_*-1), \qquad \xi \equiv \frac{q^2}{r_*
  \beta^2}. 
\]
The discriminant of this equation is given by
\[
\Delta = \tfrac{1}{4} \xi^2 + \tfrac{1}{27} \eta^3.
\]
It can be verified that $D_*-1<0$ if $r_* \geq r_+$ so that $\eta<0$. Furthermore, some
lengthy algebra shows that $\Delta>0$ for $r_*\geq r_+$. Thus, $Q(\bar{r})$
has 3 different real roots. Now, one readily sees that 
$\mbox{d}Q/\mbox{d}\bar{r}=0$ if $\bar{r}= \pm \tfrac{1}{3}\sqrt{-\eta}$.
Furthermore a computation shows $Q(\tfrac{1}{3}\sqrt{-\eta})<0$.
Hence, given that $Q(\bar{r})\rightarrow \infty$ as $\bar{r}\rightarrow
\infty$, it follows that $Q(\bar{r})$ has at least one positive
root. Moreover, as a consequence of the Descartes rules of signs one
has that $Q(\bar{r})$ has exactly 2 positive roots and one
negative.

\medskip
Let $\alpha_3$ denote the root of $Q(\bar{r})$
obtained from using Cardano's formula:
\[
\alpha_3 = \sqrt[3]{-\tfrac{1}{2}\xi + \sqrt{\Delta}} + \sqrt[3]{-\tfrac{1}{2}\xi - \sqrt{\Delta}}.
\]
It can be verified that if $r_*>r_+$ and $m^2\geq q^2$, one has
$\alpha_3>0$. The remaining two roots are given in terms of $\alpha_3$
and $\eta$ by
\[
\alpha_1 = -\tfrac{1}{2}\alpha_3 -\tfrac{1}{2} \sqrt{-3\alpha^2_3 -
  4\eta}, \qquad \alpha_2 = -\tfrac{1}{2}\alpha_3 +\tfrac{1}{2} \sqrt{-3\alpha^2_3 -
  4\eta}.
\]
Notice, in particular, that $3\alpha^2_3 +4\eta<0$. In the extremal case,
$q^2=m^2$, one obtains the simpler expressions
\begin{subequations}
\begin{eqnarray}
&& \alpha_1 =-\frac{r_*}{4(r_*-m)}\left(
  m + \sqrt{m(8r_*-7m)} \right), \label{Alpha1Extremal}\\
&& \alpha_2 = -\frac{r_*}{4(r_*-m)}\left(
  m - \sqrt{m(8 r_*-7m)} \right), \label{Alpha2Extremal}\\
&& \alpha_3 = \frac{m r_*}{2(r_*-m)}. \label{Alpha3Extremal}
\end{eqnarray}
\end{subequations}
For future reference it is noticed that given 
\[
Q_* \equiv Q(r_*) = 2r^2_* \left( 2r_*^2 -5m r_* + 3q^2\right),
\]
one has that 
\begin{eqnarray*}
Q(r_*) \geq 0 & \mbox{if} & r_*\in[r_\circledast,\infty), \\
Q(r_*) <0 & \mbox{if} & r_* \in [r_+, r_\circledast),
\end{eqnarray*}
where
\begin{equation}
r_\circledast \equiv \tfrac{5}{4}m + \tfrac{1}{4} \sqrt{25m^2 -24 q^2}.
\label{rCircledast}
\end{equation}
The constant $r_\circledast$ will be seen to play a central role in
the subsequent analysis.  In particular, one has the following lemma
obtained from lengthy calculations using the expressions obtained in
the previous paragraphs:

\begin{lemma}
\label{Lemma:BehaviourRoots}
Given $m^2> q^2$, the roots $\alpha_1$, $\alpha_2$, $\alpha_3$ of
the polynomial $Q(\bar{r})$ satisfy the inequalities
\begin{eqnarray*}
 \alpha_1 < 0 < \alpha_2 < r_- <r_+<  r_* < r_\circledast < \alpha_3 &
 \mbox{ if } &  r_*\in(r_+,r_\circledast), \\
\alpha_1 < 0 < \alpha_2 < r_- < r_+<r_\circledast < \alpha_3 <r_* & \mbox{
  if } &  r_*\in (r_\circledast, \infty)., \\
\alpha_1 <0 <\alpha_2 <r_- <r_+ <\alpha_3=r_\circledast & \mbox{ if } & r_*=r_\circledast.
\end{eqnarray*}
In the extremal case ($m^2=q^2$) one has
\begin{eqnarray*}
 \alpha_1 < 0 < \alpha_2<m <  r_* < r_\circledast < \alpha_3 &
 \mbox{ if } &  r_*\in(m,r_\circledast), \\
\alpha_1 < 0 < \alpha_2 <m< r_\circledast < \alpha_3 <r_* & \mbox{
  if } &  r_*\in (r_\circledast, \infty), \\
\alpha_1< 0 < \alpha_2 < m <\alpha_3=r_\circledast & \mbox{ if } & r_*=r_\circledast. 
\end{eqnarray*}
\end{lemma}

In terms of the roots $\alpha_1$, $\alpha_2$ and $\alpha_3$, equation \eqref{ReducedSqrt} can be conveniently rewritten as 
\begin{equation}
\bar{r}^{\prime 2} = \frac{\beta^2}{\bar{r}^2}(\bar{r}-r_*)(\bar{r}-\alpha_1)(\bar{r}-\alpha_2)(\bar{r}-\alpha_3). 
\label{ReducedSqrtRoots}
\end{equation}

\section{Analysis of the conformal curves}
\label{Section:AnalysisConformalCurves}

Our study of the behaviour of the conformal curves on the
Reissner-Nordstr\"om will be based on an analysis of the reduced
equation \eqref{ReducedEquation}. Three qualitatively different
behaviours can be identified according to whether $r_* <
r_\circledast$, $r_*=r_\circledast$ or $r_*>r_\circledast$. As it will
be seen, these
cases are associated, respectively, with a periodic, constant or
monotonically increasing behaviour of the function $\bar{r}$. 

\subsection{Conformal curves with constant $\bar{r}$}
\label{Section:CriticalCurve}
We start by investigating the possibility of having a conformal curve for
which $\bar{r}$ is constant ---that is, $\bar{r}'=\bar{r}''=0$. Using
equation \eqref{rprimeprime} one obtains the condition
\begin{equation}
\tfrac{1}{2} \partial_{\bar{r}} D(\bar{r}) = \beta \sqrt{D(\bar{r})}.
\label{Condition:ConstantCurves}
\end{equation}
The latter can be solved to give
\[
\bar{r} = \tfrac{5}{4}m \pm \tfrac{1}{4} \sqrt{25m^2 -24 q^2}. 
\]
Under the assumption $m^2\geq q^2$, it can be readily verified that
\[
\tfrac{5}{4}m - \tfrac{1}{4} \sqrt{25m^2 -24 q^2} \leq r_+ 
\leq  \tfrac{5}{4}m + \tfrac{1}{4} \sqrt{25m^2 -24 q^2} =r_{\circledast}.
\]
Thus, only the solution to condition \eqref{Condition:ConstantCurves}
with the positive radicand (i.e. $\bar{r}=r_\circledast$) will be
relevant for our subsequent discussion. For reasons which will become
clearer in the sequel, this particular conformal curve will be known
as the \emph{critical curve}.

 Some intuition can be obtained by
evaluating the constant $r_\circledast$ for the particular cases of
the Schwarzschild and the extremal Reissner-Nordstr\"om spacetime:
\begin{eqnarray*}
&& r_{\circledast} = \tfrac{5}{2}m, \qquad \mbox{ for } q=0, \\ 
&& r_{\circledast} = \tfrac{3}{2}m, \qquad \mbox{ for } q^2=m^2.
\end{eqnarray*}
Using the cases $q=0$ and $m^2=q^2$ as boundaries one can readily
check that $\tfrac{1}{9} \leq D_\circledast \leq \tfrac{1}{5}$,  where $D_\circledast \equiv D(r_\circledast)$.

\medskip
In order to understand the nature of the curve under consideration,
one has to analyse the behaviour of the functions $\bar{t}$, $\bar{u}$
and $\bar{v}$. Using equations \eqref{PhysicalNormalisation} and
\eqref{NullReducedA}-\eqref{NullReducedB} with $\bar{r}'=0$ one obtains
\begin{equation}
\bar{t} =\bar{u}-u_\circledast=\bar{v}-v_\circledast= \frac{\bar{\tau}}{\sqrt{D_{\circledast}}},
\label{SeparationCurvePhysicalParameter}
\end{equation}
where $u_\circledast\equiv u_*|_{r_*=r_\circledast}$,
$v_\circledast\equiv v_*|_{r_*=r_\circledast}$ and $u_*$ an $v_*$ are
given by expressions \eqref{NullCoordinatesDataNonExtremal} or
\eqref{NullCoordinatesDataExtremal} depending on whether one considers
the non-extremal or extremal case. Equation
\eqref{SeparationCurvePhysicalParameter} can be expressed in terms of
the unphysical proper time $\tau$ using formula
\eqref{PhysicalToUnphysicalProperTime}. One finds that
\[
{t} =
{u}-u_\circledast={v}-v_\circledast=\frac{r_\circledast}{2D_\circledast}\ln \left( \frac{2\Theta_\circledast + \beta_\circledast\tau}{2\Theta_\circledast -\beta_\circledast\tau} \right),
\]
where $\Theta_\circledast$ denotes the value of the conformal factor
of equation \eqref{ConformalFactorReduced} evaluated at $(\tau=0,r_*=r_\circledast)$.

\medskip
From the expressions in \eqref{SeparationCurvePhysicalParameter} one
has that
\[
\bar{t},\; \bar{u},\; \bar{v}\rightarrow \infty \qquad \mbox{ as } \qquad
\bar{\tau}\rightarrow \infty.
\]
The region of the Reissner-Nordstr\"om spacetime described by such
behaviour corresponds, respectively, to the lowermost  $i^+$ in the right hand side
of the conformal diagram of the non-extremal Reissner-Nordstr\"om
spacetime, and the lowermost one of
the extremal case ---see Figure \ref{ConformalDiagramRN}. This
computation also shows that close to future null infinity  the null coordinate $u$ tends
asymptotically to an affine parameter of the generators of
$\mathscr{I}^+$ ---see e.g. \cite{GriPod09,PenRin86,Ste91}. The
unphysical proper time required for the conformal curve to reach
future null infinity is given by
\[
\tau_{i^+} = \frac{2\Theta_\circledast}{\beta_\circledast}= \frac{1}{r_\circledast\sqrt{D}_\circledast}.
\]
Again, by looking at the cases $q=0$, $q^2=m^2$ one can estimate
\[
\frac{2}{\sqrt{5}m} \leq \tau_{i^+} \leq \frac{2}{m}.
\]

\subsection{Conformal curves with $r_\circledast < r_*$}
\label{Section:CGOutside}

The analysis of the conformal geodesics in the Schwarzschild spacetime
of \cite{Fri03c} proceeded by solving explicitly equation
\eqref{ReducedEquation} in terms of elliptic functions. One could
approach the analysis of the conformal geodesics in the
Reissner-Nordstr\"om spacetime in a similar fashion. Notice, however,
that while in the Schwarzschild case the polynomial $P(\bar{r})$ is
cubic, in the present case one has to deal with a quartic
polynomial. Hence, the elliptic functions one has to deal with are
more complex. In view of potential extensions of the present analysis
to more general classes of spacetimes it is desirable to follow a
procedure which, in as much as it is possible, does not depend on explicit solutions.

\medskip
If $r_\circledast<r_*$, a computation using equation
\eqref{StandardCGEqn2} together with the initial data
\eqref{CGInitialData} shows that $\bar{r}''_*>0$. As
$\bar{r}'_*=0$, one has a local minimum at $\bar{\tau}=0$ and one
needs to consider the positive root of equation
\eqref{ReducedSqrt}:
\begin{subequations}
\begin{eqnarray}
&& \bar{r}'=\sqrt{(\gamma+\beta \bar{r})^2 - D(\bar{r})}, \label{ReducedEquationExternalRegiona}\\
&&\phantom{\bar{r}'} = \frac{\beta}{\bar{r}}
\sqrt{(\bar{r}-r_*)(\bar{r}-\alpha_1)(\bar{r}-\alpha_2)(\bar{r}-\alpha_3)}.
\label{ReducedEquationExternalRegionb}
\end{eqnarray}
\end{subequations}
Using Lemma \ref{Lemma:BehaviourRoots}, it
follows that a conformal curve with $r_\circledast<\bar{r}(0)=r_*$ has no turning
points ---i.e. $\bar{r}'\neq 0$ for $\bar{\tau}>0$. Accordingly, the
function $\bar{r}$ is monotonically increasing if $r_*>r_\circledast$.

\medskip
 In order to assert the global existence of the solution to equation
\eqref{ReducedEquationExternalRegiona} ---or
\eqref{ReducedEquationExternalRegionb}--- one has to verify that
$\bar{r}$ does not blow up for finite $\bar{\tau}$. A
direct computation shows that 
\[
\partial_{\bar{r}} D(\bar{r}) = \frac{2}{\bar{r}^3}(m \bar{r}-q^2).
\]
Thus, $\partial_{\bar{r}} D(\bar{r})>0$ if $q^2/m<\bar{r} $. As it is
being assumed that $q^2
\geq m^2$ and  one has that $\tfrac{3}{2}m\leq r_\circledast$, it
follows, indeed, that $\partial_{\bar{r}}
D(\bar{r})>0$. Hence, one has $0<D_* \leq D(\bar{r})$. Furthermore,
one finds that
\[
\bar{r}' \leq \sqrt{(\gamma + \beta \bar{r})^2-D_*} < |\gamma + \beta \bar{r}|.
\]
Now, for $r_* \leq \bar{r}$, one has that
\[
|\gamma + \beta \bar{r}| = D_*\left| 1- 2\frac{\bar{r}}{r_*} \right| =D_*\left(2\frac{\bar{r}}{r_*}-1\right).
\]
Thus, one deduces to the differential inequality
\[
\bar{r}' < \frac{2D_*}{r_*}\bar{r}
\]
which can be integrated to give
\[
\bar{r} \leq r_* e^{2D_*\bar{\tau}/r_*}.
\]
Thus, one has that $\bar{r}$ does not blow up for finite value of
$\bar{\tau}$, and the solution to
\eqref{ReducedEquationExternalRegiona} with $\bar{r}(0)=r_*>r_\circledast$ exists
for all times. 

\medskip
In order to obtain more information about the solution we construct a
lower fence. From \eqref{ReducedEquationExternalRegionb} one readily
has that
\begin{equation}
 \frac{\varkappa\beta}{\bar{r}}\sqrt{\bar{r}-r_*} <\bar{r}' 
\label{LowerBarrier}
\end{equation}
where
\[
\varkappa \equiv \sqrt{(r_*-\alpha_1)(r_*-\alpha_2)(r_*-\alpha_3)}.
\]
The inequality \eqref{LowerBarrier} can be integrated to yield
\[
r_* +\left( \tfrac{1}{2}\left( 6\varkappa \beta\bar{\tau}+2\sqrt{16 r_*^3
      +9\varkappa^2\beta^2 \bar{\tau}^2} \right)^{1/3} - \frac{2 r_*}{\left(6\varkappa\beta^2 \bar{\tau}+\sqrt{16 r_*^3
      +9\varkappa^2 \beta^2\bar{\tau}^2}\right)^{1/3}}\right)^2 < \bar{r}.
\]
From this last inequality it follows that $\bar{r}\rightarrow \infty$
as $\bar{\tau}\rightarrow \infty$. Recall that $\bar{\tau}\rightarrow
\infty$ corresponds, following formula
\eqref{PhysicalToUnphysicalProperTime}, to a finite value of
$\tau$. Accordingly, the conformal curve reaches future null infinity
$\mathscr{I}^+$ for a finite value of $\tau$. 

\medskip
In order discuss the behaviour of the function $\bar{t}$, we consider
the normalisation condition \eqref{PhysicalNormalisation} with
$\bar{r}$ being the solution discussed in the previous paragraphs. An
argument similar to the one used for $\bar{r}$ shows that $\bar{t}$
with $r_*>r_\circledast$ does not blow up in finite time and that
$\bar{t}\rightarrow \infty$ as $\bar{\tau}\rightarrow \infty$. The
coordinate $t$, however, is not a good coordinate to discuss the
behaviour of the conformal curves with respect to (future) null
infinity. In order to do this, we consider
equations \eqref{NullReducedA}, with $\bar{r}$ the
solution of equation \eqref{ReducedEquationExternalRegiona} with
$r_*>r_\circledast$. Global existence of solutions to
\eqref{ReducedEquationExternalRegiona} follows, again, by showing that
the solution (and its derivative) does not blow up in finite time. We
now look in more detail at the behaviour of $\bar{u}$ as
$\bar{\tau}\rightarrow\infty$. As $D(\bar{r})$ is bounded for
$\bar{r}\in (r_\circledast,\infty)$, it follows that there exists a
sufficiently large positive number $\bar{\tau}_{\times}$ and positive constants, $C_1$
and $C_2$, for which
\[ 
\frac{C_1}{\bar{r}'}<\bar{u}' < \frac{C_2}{\bar{r}'}, \qquad
\mbox{for} \qquad \bar{\tau}>\bar{\tau}_{\times}.
\]
Using the chain rule in the form $\bar{u}' = \bar{r}'
\mbox{d}u/\mbox{d}\bar{r}$ and by increasing $\bar{\tau}_{\times}$ if
necessary, one can find constants, $\tilde{C}_1$ and $\tilde{C_2}$,
for which one has
\[
\frac{\tilde{C}_1}{\bar{r}^2} <
\frac{\mbox{d}\bar{u}}{\mbox{d}\bar{r}} <
\frac{\tilde{C}_2}{\bar{r}^2}, \qquad \mbox{for} \qquad \bar{r} > \bar{r}(\bar{\tau}_{\times}).
\]
From the latter, it follows that $\bar{u}$ goes to a finite constant
value (which depends on $r_*$) as $\bar{\tau}\rightarrow
\infty$. Thus, one concludes that in the unphysical (i.e. conformally
rescaled) picture, the conformal curve reaches future null infinity
for a finite value of the unphysical proper time $\tau_{\mathscr{I}^+}$. This value
can be read from the conformal factor \eqref{ConformalFactorReduced}
to be $\tau_{\mathscr{I}^+} = 2\Theta_*/\beta$. 

\medskip
We summarise the results of the present section in the following
proposition:

\begin{proposition}
The conformal curves with initial data given by \eqref{CGInitialData} and
$r_*>r_\circledast$ exist for all $\bar{\tau}\in[0,\infty)$. The
curves reach future null infinity for a finite value of the parameter
$\tau$. 
\end{proposition}

\subsection{Conformal curves with $r_* < r_\circledast$}
\label{Section:CGInside}

The case of conformal curves with $r_*<r_\circledast$ is, in some
sense, the most interesting one. From equation \eqref{ReducedSqrt},
the factorisation of the quartic polynomial $P(\bar{r})$ and Lemma
\ref{Lemma:BehaviourRoots} one concludes that if $r_*<r_\circledast$,
then the resulting conformal curve will have turning points at
$\bar{r}=r_*$ and $\bar{r}=\alpha_2$. Using equation
\eqref{rprimeprime} a computation shows that
\[
\bar{r}''|_{\bar{r}=r_*} <0, \qquad \bar{r}''|_{\bar{r}=\alpha_2}>0. 
\]
Hence, the turning points given by $r_*$ and $\alpha_2$ correspond to
a maximum and a minimum of the function $\bar{r}$, respectively. As
$\bar{r}$ (and $\bar{r}'$) are bounded, one concludes that the
solution to equation \eqref{ReducedSqrt} with $r_*<r_\circledast$
exists for all $\bar{\tau}>0$.

\subsubsection{Behaviour of the curves in the non-extremal case}

 As a consequence of the discussion in the previous paragraph, the
 function $\bar{r}$ is initially decreasing. Moreover, from equations
 \eqref{NullReducedA}-\eqref{NullReducedB} one has that $u'_*,\,v'_*>0$
 so that the concavity of the curve points, initially, away from the
 horizon ---see figure \ref{Figure:Curves}.
From $\alpha_2<r_+$, it
 follows there
 exists a value of $\bar{\tau}$ for which $\bar{r}=r_+$ ---implying
 that the conformal curve crosses the horizon. The temporal coordinate
 $t$ is not appropriate for this discussion as $t\rightarrow \infty$ for any
 curve approaching the event horizon ---this can readily be seen using the
 normalisation \eqref{PhysicalNormalisation} and the fact that
 $D(r)=0$ at the event horizon ---that is, at $r=r_+$. Hence, one makes use of the null coordinate
 $v$. The evolution of this coordinate along the conformal curve is described by
 equation \eqref{NullReducedB}.  Initially, one has
 that $\bar{v}'(0)=1/\sqrt{D(\bar{r})}$, so that $\bar{v}$ is, at least initially,
 increasing. As a
consequence of Lemma \ref{Lemma:BehaviourRoots} one sees that the
conformal curve should cross the event and the Cauchy horizon before
reaching the minimum of $\bar{r}$ at $\bar{r}=\alpha_2$. It is just
necessary to check that the various field remain regular at the horizons.

\medskip
 As the curve approaches the event horizon at
 $\bar{r}=r_+$, both the numerator and denominator of the right hand
 side of equation \eqref{NullReducedB} vanish ---in particular, the numerator can
 vanish as $\bar{r}'<0$. Using the L'Hopital rule and equation
\eqref{rprimeprime} one finds that $\bar{v}'|_{\bar{r}=r_+}$ is well
defined and positive. Accordingly, the curve enters region II of
conformal diagram of the non-extremal Reissner-Nordstr\"om spacetime
---see Figure \ref{Figure:Curves}. A similar situation occurs as
$\bar{r}$ approaches $r_-$: the numerator and denominator of equation
\eqref{NullReducedB} both vanish; using the L'Hopital on verifies that
$\bar{v}'$ is well defined at $r_-$. Notice that as $\bar{r}'<0$ and
$\bar{v}'>0$, the curve exits the region II of the conformal diagram
through the left hand side of the Cauchy horizon. The turning point
at $\bar{r}=\alpha_2$ is located in the region III of the conformal
diagram. After the curve has reached this point, one has that
$\bar{r}'>0$ and the behaviour of the curve as it approaches again the
horizon at $\bar{r}=r_-$ is different. In this case the numerator of
the right hand side of equation \eqref{NullReducedB} tends to a
non-zero value and one has that $\bar{v}\rightarrow \infty$ ---and
hence, also $\bar{v}\rightarrow \infty$. In order to discuss the
behaviour curve beyond this point one would have to introduce a new
set of null coordinates.

\medskip
In subsection \ref{ExplicitExpressions} it will be shown that the
points with $\bar{r}=r_+, \; r_-,\; \alpha_2$ are reached for a
finite value of the physical proper time $\bar{\tau}$. 

\begin{figure}[t]
\centerline{\includegraphics[width=\textwidth]{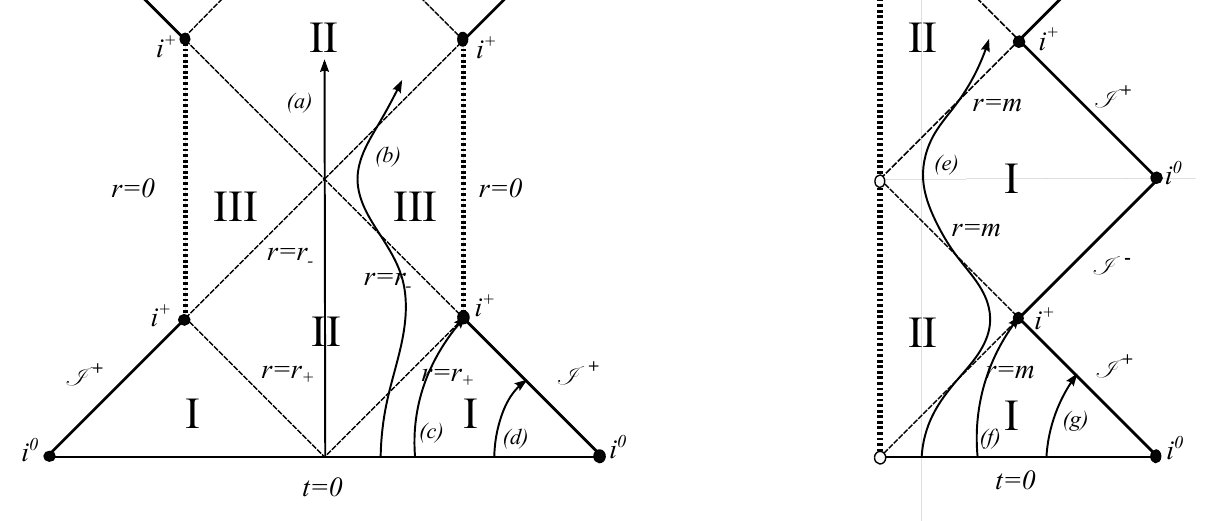}}
\caption{Schematic illustration of the behaviour of conformal
curves. To the left the non-extremal case: (a) the curve starting at
$r_*=r_+$; (b) a curve with $r_*<r_\circledast$; (c) the critical
curve; (d) a curve with $r_*>r_\circledast$. To the right the extremal
case: (e) a curve with $r_*<r_\circledast$; (f) the critical curve;
(g) a curve with $r_*>r_\circledast$. \emph{The curves are not
depicted on scale.}}
\label{Figure:Curves}
\end{figure}

\subsubsection{Behaviour of the curves in the extremal case}

As in the non-extremal case, conformal
curves in the extremal Reissner-Nordstr\"om spacetime with
$r_*<r_\circledast$ satisfy $\bar{r}'_*<0$ and $\bar{v}'_*>0$. Using
the L'Hopital rule one can verify that $\bar{v}'$ is well defined and
positive at $\bar{r}=r_+=m$. Thus, the conformal curve penetrates in the
region II of the conformal diagram of the spacetime ---cf. Figure \ref{Figure:Curves}. The essential difference
with respect to the non-extremal case is that the turning point
given by $\bar{r}=\alpha_2$ is now in region II. After the curve has
passed this point one has $\bar{r}'>0$ and $\bar{v}'>0$. Hence, from
equation \eqref{NullReducedB} it follows that $
\bar{v},\; \bar{v}'\rightarrow \infty$ the second time the curve
approaches $\bar{r}=r_+=m$. In order to follow the behaviour of the curve
beyond this point one would need a new set of null coordinates.

\medskip
In subsection \ref{ExplicitExpressions} it will be shown that the
points with $\bar{r}=m,\; \alpha_2$ are reached for a
finite value of the physical proper time $\bar{\tau}$.

\subsubsection{Regions of the spacetime not covered by the congruence}
As already discussed, the turning point described by the condition
$\bar{r}=\alpha_2$ is located in the region II in the extremal
case and in region III in the non-extremal case. In these regions the
curves $\bar{r}=\mbox{constant}$ are timelike. Regarding $\alpha_2$ as
a function of $r_*$, it follows from the
expressions given in Section \ref{Section:PolynomialP}  that there exists a certain value of
$r_*$ for which $\alpha_2$ attains a (non-zero) minimum. This minimum
value depends only on the value of $m$ and $q$. 

\medskip
 The phenomenon described in the previous paragraph is better analysed in the extremal case where simpler
analytical expressions are available. Using expression
\eqref{Alpha2Extremal} one can readily see that $\alpha_2$ attains a
minimum value of $(\sqrt{2}-1/2)m\approx 0.91m$ along the conformal
curve with $r_* = (2-1/\sqrt{2})m\approx 1.29m$. All other conformal
curves with $r_*<r_\circledast$ will have a higher value for the
$r$-location of the turning point.

\medskip
The discussion in the previous paragraphs shows that there exist
regions in the regions III of the non-extremal Reissner-Nordstr\"om
spacetime and the regions II in the extremal case that cannot be probed
by means of the family of conformal curves under consideration. \emph{In
particular, a conformal curve cannot get arbitrarily close to the
singularities of the spacetime ($r=0$). In this sense, one can regard
our class of conformal curves as singularity avoiding.}

\subsubsection{Explicit expressions in terms of elliptic functions}
\label{ExplicitExpressions}

As discussed previously, the function $\bar{r}$ is decreasing if $\bar{r}\in(\alpha_2,r_*)$. If this is the case then equation \eqref{ReducedSqrt} implies  
\[
\bar{r}'=-\frac{\beta}{\bar{r}}
\sqrt{(\bar{r}-r_*)(\bar{r}-\alpha_1)(\bar{r}-\alpha_2)(\bar{r}-\alpha_3)}.
\]
The latter implies, in turn
\[
\bar{\tau} =- \frac{1}{\beta} \int_{r_*}^{\bar{r}}  \frac{s\mbox{d}s}{\sqrt{(s-\alpha_1)(s-\alpha_2)(s-\bar{r}_*)(s-\alpha_3)}}.
\]
The integral in the right hand side can be evaluated in terms of elliptic functions ---see e.g. \cite{Law89}. For example, the physical proper time required by the curve to go from $\bar{r}=r_*$ to $\bar{r}=\alpha_2$ is given by 

\begin{eqnarray*}
&& \beta \bar{\tau}(\alpha_2) = \frac{2 \alpha_1}{\sqrt{(\alpha_3-\alpha_2)(\bar{r}_*-\alpha_1)}}K\left(\sqrt{\frac{(\bar{r}_*-\alpha_2)(\alpha_3-\alpha_1)}{(\alpha_3-\alpha_2)(\bar{r}_*-\alpha_1)}}\right) \\
&& \hspace{2cm} + \frac{2(\alpha_2-\alpha_1)}{\sqrt{(\alpha_3-\alpha_2)(\bar{r}_*-\alpha_1)}}\Pi\left( \frac{\bar{r}_*-\alpha_2}{\bar{r}_*-\alpha_1}; \sqrt{\frac{(\bar{r}_*-\alpha_2)(\alpha_3-\alpha_1)}{(\alpha_3-\alpha_2)(\bar{r}_*-\alpha_1)}} \right),
\end{eqnarray*}
where $K(\;\cdot\;)$ and $\Pi(\;\cdot\;;\;\cdot\;)$ denote, respectively, the
complete elliptic integrals of the first and third kind. The above
expression is valid for both the non-extremal and extremal cases. For
specific values of $m$ and $\bar{r}_*$, these integrals can be
accurately evaluated with a computer algebra system. Hence, the
expression for $\bar{\tau}(\alpha_2)$ can be used to verify the
accuracy of numerical solutions to the conformal curve equations. However, for analytical purposes,
expressions of the type given above are too clumsy to be used. Insight
into the behaviour of $\bar{\tau}$ regarded as a function of $\bar{r}$
and $r_*$ can be obtained by means of suitable estimates.

\medskip
As a consequence of Lemma \ref{Lemma:BehaviourRoots} one has that for $\alpha_2\leq \bar{r}\leq r_*$ one has
\[
-\varkappa_1^2 (\bar{r}-\alpha_2)(\bar{r}-r_*) \leq P(\bar{r}) \leq -\varkappa_2^2 (\bar{r} -\alpha_2)(\bar{r}-r_*),
\]
where
\[
\varkappa^2_1 \equiv (\alpha_2-\alpha_1)(\alpha_3-\alpha_2), \qquad \varkappa^2_2 \equiv (r_*-\alpha_1)(\alpha_3-r_*). 
\]
Furthermore, one finds that
\[
\frac{1}{\varkappa_1} I(\bar{r}) \leq \beta \bar{\tau}\leq \frac{1}{\varkappa_2} I(\bar{r}),
\]
where 
\begin{eqnarray*}
&& I(\bar{r}) \equiv \int_{\bar{r}}^{r_*} \frac{s \mbox{d}s}{\sqrt{(s-\alpha_2)((r_*-s))}}, \\
&& \phantom{I(\bar{r})}= \sqrt{(\bar{r}-\alpha_2)(r_*-\bar{r})}+\tfrac{1}{2}(\alpha_2+r_*) \left( \arcsin\left( \frac{3r_*+\alpha_2}{(\alpha_2-r_*)^2} \right)-\arcsin\left(\frac{2\bar{r}+r_* +\alpha_2}{(\alpha_2-r_*)^2}\right)\right).
\end{eqnarray*}
In particular, one has that $I(\alpha_2)$ is finite. Thus, one
concludes that the conformal curves with $r_*<r_\circledast$ reach the
turning point $\bar{r}=\alpha_2$ (and hence, also the horizons at $\bar{r}_\pm$) in finite physical proper time.

\subsection{The conformal curve starting at $r_+$}
\label{Section:CGBifurcationSphere}

In the case of a non-extremal Reissner-Nordstr\"om spacetime it is also of interest to analyse the behaviour
of the conformal curve starting at the bifurcation sphere
($\bar{r}_*=\bar{r}_+$). Observing that $D(r_+)=0$ so that
$\beta_+=0$, it follows from equation \eqref{rprimeprime} that
\[
\bar{r}'' = \frac{q^2}{\bar{r}^3}-\frac{m}{\bar{r}^2}.
\]
Using the initial conditions 
\[
 r_*= r_+, \qquad \bar{r}_*' =0,
\]
one can readily integrate to obtain
\[
\bar{r} = m + \sqrt{m^2-q^2} \cos \bar{\tau}>0.
\] 
Because of the reflection symmetry of the Reissner-Nordstr\"om
spacetime with respect to the bifurcation sphere and the timelike
character of the curve in question, the conformal curve remains always
remains in the middle of the conformal diagram of the spacetime. The
curve starts in region II where the surfaces of constant $\bar{r}$ are
spacelike. The function $\bar{r}$ is decreasing in the region II, and
reaches its minimum value at the Cauchy horizon ---where
$\bar{r}=\bar{r}_-$. From here $\bar{r}$ becomes increasing and
eventually reaches its maximum (and initial value) at
$\bar{r}=\bar{r}_+$. This corresponds to a second bifurcation sphere
in the Penrose diagram of the maximal analytic extension of the
spacetime ---see Figure \ref{Figure:Curves}. It is worth observing
that any point along the conformal curve can be reached in a finite
amount of (physical or unphysical) proper time. In particular, the
distance between two consecutive bifurcation spheres in the Penrose
diagram measured in terms of the parameter $\bar{\tau}$ is $2\pi$.

\subsection{Further analysis of the behaviour of the critical curve
  close to timelike infinity}
\label{Section:TimelikeInfinity}

The purpose of the present section is to further discuss the behaviour
of the critical curve $\bar{r}=r_\circledast$ as it approaches
$i^+$. The reason for this analysis is motivated by the observation
made in \cite{Fri03a} that in the Schwarzschild spacetime the
corresponding critical curve, which is timelike for
$\tau\in[0,\tau_{i^+})$, becomes null at $\tau=\tau_{i^+}$. This
observation indicates a degeneracy of the conformal structure of the
spacetime.

\subsubsection{The intersection of the critical curve and null infinity}
First, we consider the behaviour of conformal curves with
$r_*>r_\circledast$ as $r_*\rightarrow r_\circledast$ from the
right. These curves reach future null infinity in
a finite amount of unphysical proper time
\[
\tau_{\mathscr{I}^+}\equiv \frac{2 \Theta_*}{\beta} = \frac{1}{r_* \sqrt{D_*}},
\] 
as it can be seen from the expression \eqref{ConformalFactorReduced}
for the conformal factor $\Theta$.  Setting $r_*=(1+\epsilon)r_\circledast$, 
for small $\epsilon>0$, it follows that
\[
\tau_{\mathscr{I}^+} = \frac{1}{r_\circledast
  \sqrt{D_\circledast}}\left(1+3\epsilon +\mathcal{O}(\epsilon^2)\right).
\]
Thus, one has in particular that
\[
\frac{\mbox{d}\tau_{\mathscr{I}^+}}{\mbox{d}\epsilon}\big|_{\epsilon=0}=\frac{3}{r_\circledast
  \sqrt{D_\circledast}},
\]
that is, future null infinity approaches $i^+$ at a finite positive angle.

\subsubsection{The intersection of the critical curve and the horizon}
The analysis of the behaviour of the conformal curves with $r_*<r_\circledast$ is more
delicate.  In what follows, let $\bar{\tau}_{\mathscr{H}^+}$ denote
the value of the physical proper time for which a conformal curve
reaches the horizon. From the expression \eqref{ReducedSqrtRoots} it
follows that
\[
\beta\bar{\tau}_{\mathscr{H}^+} = \int_{r_+}^{r_*} \frac{s \mbox{d}s}{\sqrt{(s-r_*)(s-\alpha_1)(s-\alpha_2)(s-\alpha_3)}}.
\]
Now, setting $r_* = (1-\epsilon)r_\circledast$
for small $\epsilon>0$, and observing  Lemma
\ref{Lemma:BehaviourRoots} it follows that
\[
 \frac{s}{\sqrt{(s-r_*)(s-\alpha_1)(s-\alpha_2)(s-\alpha_3)}}=\frac{s}{(r_\circledast
   -s)\sqrt{(s-\bar{\alpha}_1)(s-\bar{\alpha}_2)}} + \mathcal{O}(\epsilon), \qquad s<r_\circledast,
\]
where $\bar{\alpha}_1$ and $\bar{\alpha}_2$ are the values of the
roots $\alpha_1$ and $\alpha_2$ corresponding to $r_* =
r_\circledast$.  A computation then shows that
\[
\beta\bar{\tau}_{\mathscr{H}^+} =\kappa -
\frac{r_\circledast}{\sqrt{(r_\circledast-\bar{\alpha}_1)(r_\circledast-\bar{\alpha}_2)}}\ln\epsilon
+ \mathcal{O}(\epsilon),
\]
with $\kappa$ a constant depending on $m$ and $q$. One readily sees
that $\bar{\tau}_{\mathscr{H}^+}\rightarrow \infty$ as
$\epsilon\rightarrow 0$ consistently with the discussion of Section
\ref{Section:CriticalCurve}. It can be explicitly shown that
\[
\eta \equiv \frac{r_\circledast}{\sqrt{(r_\circledast-\bar{\alpha}_1)(r_\circledast-\bar{\alpha}_2)}}\leq 1,
\]
where the equality is achieved only for $q^2=m^2$. Finally, noticing that
\[
\frac{2\Theta_*}{\beta}= \frac{1}{r_\circledast \sqrt{D_\circledast}}\left( 1+
  \frac{1}{r_\circledast D_\circledast}(r_\circledast -m)\epsilon + \mathcal{O}(\epsilon^2) \right),
\]
one concludes from equation \eqref{UnphysicalToPhysicalProperTime} that
\[
\tau_{\mathscr{H}^+} = \frac{1}{r_\circledast
  \sqrt{D_\circledast}}\left(1+2
  e^{-\kappa}\epsilon^\eta
  + \frac{1}{r_\circledast D_\circledast}(r_\circledast-m)\epsilon +
  \mathcal{O}\left(\epsilon^{1+\eta}\right)\right). 
\]

It follows that
\[
\lim_{\epsilon\rightarrow 0}\frac{\mbox{d} \tau_{\mathscr{H}^+}}{\mbox{d}\epsilon} =
\left\{
\begin{array}{lll}
\infty && q^2<m^2 \\
&& \\
\displaystyle\frac{2}{m}(5+2e^{-\kappa})<\infty && q^2=m^2.
\end{array}
\right. .
\]
Thus, the critical conformal curve and the horizon are tangent at $i^+$
for $q^2<m^2$ ---that is, the critical curve becomes null at
$i^+$. A similar singular behaviour for the conformal geodesic reaching $i^+$
in the Schwarzschild spacetime has been described in
\cite{Fri03a}.  Indeed, it can be verified that
$\eta=\frac{1}{\sqrt{2}}$ for $q=0$. This behaviour indicates a degeneracy of the
conformal structure at $i^+$ for $q^2<m^2$. The most remarkable
feature of the present analysis is the fact that the critical curve remains
timelike in the extremal case $q^2=m^2$. In this case, a
more regular behaviour of the conformal structure is to be expected.

\section{Analysis of the conformal deviation equations}
\label{Section:ConformalDeviationEquations}

The purpose of this section is to show that the congruence of
conformal curves under consideration does not form caustics in the
outer domain of communication of the Reissner-Nordstr\"om
spacetime. 

\subsection{Basic equations}
Inspired on a similar discussion in \cite{Fri03c}, we will study
solutions to 
the reduced $\tilde{\bm g}$-adapted deviation equation 
\eqref{ReducedWarpProductDeviationEquation} with suitable initial
data. In the previous section the value of the coordinate $r$ on the
initial hypersurface $\tilde{S}$ has been used to parametrise the
conformal curves of our congruence. Thus, it would be natural to use
the vector ${\bm\partial}_r$ as deviation vector. However, as we are
specifically interested in the behaviour of the congruence near the
event horizon, this choice is no longer adequate. Instead, we consider the vector
field ${\bm \partial}_\varrho$ where $\varrho$ is the isotropic radial
coordinate given by formulae \eqref{IsotropicCoordinates}. Thus, in
what follows we set
\begin{equation}
\bar{\bm Z}_* = ({\bm \partial}_\varrho \bar{x})_*, \qquad \chi \equiv \tfrac{1}{2}R[{\bm l}],
\qquad \zeta \equiv \not{\!\!D}_{\bar{\bm Z}} \beta.
\label{Definitions:Deviation}
\end{equation}
It follows that equation
\eqref{ReducedWarpProductDeviationEquation} can be rewritten in the
form
\begin{equation}
\omega'' -(\beta^2 + \bar{\chi}) \omega = \zeta
\label{enemy}
\end{equation}
where the bar over $\chi$ indicates that the function is regarded
as depending on $\bar{\tau}$. In the particular case of the Reissner-Nordstr\"om spacetime one has
that
\begin{equation}
\bar{\chi} = \frac{2m}{\bar{r}^3}-\frac{3q^2}{\bar{r}^4}. 
\label{FormulaPsi}
\end{equation}
One readily sees that $\bar{\chi}>0$ if $\bar{r}>3q^2/2m$. By direct
evaluation it can be checked that if $m^2\geq q^2$, then $3q^2/2m\leq r_\circledast$, with
the equality being achieved in the extremal case. Finally, it is
noticed that 
\begin{equation}
\frac{\mbox{d}\bar{\chi}}{\mbox{d}\bar{r}} = \frac{6}{\bar{r}^5}(2q^2-m\bar{r}) >0
\qquad \mbox{ if }\quad  \bar{r}<\frac{2q^2}{m}.
\label{SignDerivativeChi}
\end{equation}

\medskip
Now, recalling that $\bar{X}= \bar{x}'$ and taking into account
\eqref{Definitions:Deviation}, one finds that the initial data for equation \eqref{enemy} on $\tilde{S}$ is given,  for $\varrho_*>\varrho_+$ (i.e
$r_*>r_+$), by
\begin{subequations}
\begin{eqnarray}
&& \omega_*\equiv {\bm \epsilon}_{\bm l}(\bar{\bm X},\bar{\bm Z})_*
=\bar{t}'_*
\left(\frac{\partial\bar{r}}{\partial\varrho}\right)_*-\bar{r}'_*\left(\frac{\partial \bar{t}}{\partial\varrho}\right)_*=\frac{r_*}{\varrho_*}>1, \label{DeviationData1}\\
&&\omega'_*=\left( {\not\!\!D}_{\bar{\bm X}} {\bm \epsilon}_{\bm
    l}(\bar{\bm X},\bar{\bm Z})\right)_*=0, \label{DeviationData2}
\end{eqnarray}
\end{subequations}
where equations \eqref{IsotropicCoordinatesDE} and \eqref{CGInitialData} have been used to
simplify the expression for $\omega_*$. 

\medskip
 Following the discussion of Section
\ref{Section:ConformalCurves}, the coefficients $\beta^2$ and $\zeta$
in equation \eqref{enemy} are constant along a given
conformal curve, and thus, can be conveniently be evaluated at the
initial hypersurface $\tilde{\mathcal{S}}$ ---cf. the remark after equation
\eqref{ConstantsAlongCurve}. Recalling that $\beta^2 = 4 D_*/r_*^2$,
it follows that
\begin{eqnarray*}
&&\partial_{\varrho_*} \beta^2 =
\left(\frac{\mbox{d}r}{\mbox{d}\varrho}\right)_* \partial_{r_*}\beta^2 =
\frac{4}{r_*^3}\left(\frac{\mbox{d}r}{\mbox{d}\varrho}\right)_* \left( r_*
  \frac{\mbox{d}D_*}{\mbox{d}r_*}-2D_* \right), \\
&& \phantom{\partial_{\varrho_*} \beta^2}= -\frac{8}{r_*^5}\left(\frac{\mbox{d}r}{\mbox{d}\varrho}\right)_* (r_*^2-3mr_*+2q^2).
\end{eqnarray*}
Finally, using $\mbox{d}r/\mbox{d}\varrho = r \sqrt{D}/\varrho$ and
$\partial_\varrho \beta = (\partial_\varrho \beta^2)/2\beta $
one concludes that
\[
\zeta = (\partial_\varrho \beta)_*=- \frac{2}{\varrho_* r_*^3} (r_*^2-3mr_*+2q^2).
\]
This expression is positive if 
\[
r_*\in \left(\tfrac{3}{2}m -
\tfrac{1}{2}\sqrt{9m^2-8q^2},\tfrac{3}{2}m + \tfrac{1}{2}\sqrt{9m^2-8q^2}\right).
\]
In the extremal case the expression for $\zeta$ simplifies to
\begin{equation}
\zeta = -\frac{2}{r_*^3}(r_*-2m),
\label{ZetaExtremal}
\end{equation}
which is positive for $r_*<2m$. It is also noticed that
\[
\beta^2 + \chi_* = \frac{1}{\bar{r}_*^4}\left(4r^2_* - 6m r_* +q^2 \right).
\]
One concludes then that 
\begin{subequations}
\begin{eqnarray}
&& \beta^2+\chi_* <0 \qquad \mbox{ if } r_* \in (r_+,\tfrac{3}{4}m + \tfrac{1}{4}\sqrt{9m^2-4q^2}), \label{PositivityConditions1}\\
&& \beta^2 + \chi_* \geq 0 \qquad \mbox{ if } r_* \in
\big[\tfrac{3}{4}m + \tfrac{1}{4}\sqrt{9m^2-4q^2}, \infty
\big). \label{PositivityConditions2}
\end{eqnarray}
\end{subequations}
In the extremal case the above expressions reduce to
\begin{eqnarray*}
&& \beta^2+\chi_* <0 \qquad \mbox{ if } r_* \in \big(m,\tfrac{1}{4}(3+\sqrt{5})m\big), \\
&& \beta^2 + \chi_* \geq 0 \qquad \mbox{ if } r_* \in \big[ \tfrac{1}{4}(3+\sqrt{5})m,  \infty),
\end{eqnarray*}
where $\tfrac{1}{4}(3+\sqrt{5})\approx 1.309$. Finally, one has that
\[
\zeta + (\beta^2 + \chi_*) \omega_* = \frac{1}{\varrho_* r_*}(2r^2_* - 3q^2),
\]
so that $\omega''_*\geq 0$ if $r_* \geq \sqrt{\tfrac{3}{2}}|q|$ where
$\sqrt{\tfrac{3}{2}}\approx 1.225$.

\subsection{The curve deviation equation along the critical curve}
\label{Section:DeviationCriticalCurve}
The simplest situation on which to analyse the solutions of the
deviation equation \eqref{enemy} is  along the critical curve
$\bar{r}=r_\circledast$. Letting $\chi_\circledast \equiv
\chi(r_\circledast)$, a direct computation shows that
\begin{eqnarray*}
\chi_\circledast >0, & \zeta_\circledast>0, & \mbox{ for } m^2>q^2, \\
\chi_\circledast=0, & \zeta_\circledast>0, & \mbox{ for } m^2=q^2.
\end{eqnarray*}
With this information at hand, the solution to the deviation equation
can be found explicitly to be given by
\[
\omega(\bar{\tau}) =\left( \omega_\circledast +
  \frac{\zeta_\circledast}{\beta^2_\circledast + \chi_\circledast}
\right)\cosh\left(\sqrt{\beta^2_\circledast + \chi_\circledast}\bar{\tau}
\right) -  \frac{\zeta_\circledast}{\beta^2_\circledast + \chi_\circledast}. 
\]
This solution satisfies $\omega(\bar{\tau})>0$, $\omega'(\bar{\tau})>0$ for all $\bar{\tau}\geq
0$. Thus, no conjugate points arise in the critical
curve when regarded as a curve on the physical Reissner-Norsdstr\"om
spacetime $(\tilde{\mathcal{M}},\tilde{\bm g})$.  In view of future
applications it is important to verify the absence of conjugate points
in the conformally rescaled spacetime even at $i^+$. In order to do
this, one has to consider the conformal curve as parametrised by the
unphysical proper time $\tau$. Following the discussion of Section
\ref{Section:PhysicalMetricAdaptedEquations} one has that $\bar{\bm X} = \Theta
\dot{\bm x}$. The relevant deviation vector field is then given by
\[
\partial_\varrho {\bm x} = \partial_\varrho \bar{\bm x} -\dot{\bm
  x} \partial_\rho \tau.
\]
It follows that for $\bar{\bm X}$ and $\bar{Z}$ to
remain linearly independent one requires that $\Theta {\bm
  \epsilon}_{\bm l}(\bar{\bm X},\bar{\bm Z})\neq 0$ along the critical
curve up to (and including) $i^+$. Using expression
\eqref{ConformalFactorReducedPhysical} together with the explicit
expression for $\omega(\bar{\tau})$ obtained in the previous paragraph
one obtains that 
\[
\Theta {\bm \epsilon}_{\bm l}(\bar{\bm X},\bar{\bm Z})\geq
\Theta_\circledast\omega_\circledast>0 
\]
along the critical curve up to, and including $i^+$ so that also no
conjugate points arise on the conformal Reissner-Nordstr\"om manifold.

\subsection{Curves with  $r_\circledast <r_*$}
\label{Subsection:DeviationOutside}
The analysis of solutions of equation \eqref{enemy} with 
for $r_\circledast <r_*$ follows closely the discussion of
\cite{Fri03c}. It is given here for the sake of completeness. The key observation is that the solution to \eqref{enemy} admits the
representation
\begin{equation}
\omega(\bar{\tau}) = \varpi(\bar{\tau}) \left( \omega_* + \left(1-\frac{\varsigma(\bar{\tau})}{\varpi(\bar{\tau})}\right)\zeta \right),
\label{Solutionw}
\end{equation}
where $\varpi(\bar{\tau})$ and $\varsigma(\bar{\tau})$ are the solutions to the auxiliary problems
\begin{eqnarray*}
&& \varpi'' -(\beta^2 +\bar{\chi}) \varpi = 0, \quad \varpi(0)=1, \quad \varpi'(0)=0, \\
&& \varsigma'' -(\beta^2 +\bar{\chi}) \varsigma=-1, \quad \varsigma(0)=1, \quad \varsigma'(0)=0. 
\end{eqnarray*}
That $\omega$ as given by \eqref{Solutionw} is indeed a solution of
\eqref{enemy} with the right initial data can be verified by directed
evaluation. As $r_*>r_\circledast$, it follows from the observation
after equation \eqref{FormulaPsi} that $\beta^2+\bar{\chi}>0$ along the
conformal curves under consideration.

\medskip
Using the equation for $\varpi$ and the initial data $\varpi(0)=1$ one sees
that $\varpi''(0)>0$ so that $\varpi$ has a local minimum at
$\bar{\tau}=0$. Thus, at least for positive values of $\bar{\tau}$
close to $0$ one has that $\varpi$ must be increasing. Furthermore, as
$\bar{\chi}>0$ for $r>r_\circledast$, one finds that $\beta^2 \varpi \leq \varpi''$. This
last differential inequality can be integrated to yield $\varpi\geq \cosh(\beta
\bar{\tau})$. One concludes that $\varpi$ is increasing for all
$\bar{\tau}\geq 0$. A similar argument with the function $\eta\equiv \varpi-\varsigma$ satisfying the equation
\[
\eta''-(\beta^2+\bar{\chi}) \eta =1, \qquad \eta(0)=0, \qquad \eta'(0)=0,
\]
shows that $\varpi \geq \varsigma$ for all $\bar{\tau}\geq 0$. 

\medskip
The information obtained in the previous paragraph will be used to
estimate the term $1-\varsigma/\varpi$ in \eqref{Solutionw}. Due to the
monotonicity of $\varpi$ one has that $1-\varsigma/\varpi \geq 0$. Using that $\varpi\neq 0$
for $\bar{\tau}\geq 0$, it follows that there exists a function
$f(\bar{\tau})$ such that $\varsigma=f \varpi$. It can be readily seen that $f$
satisfies the equation
\[
f'' + 2 \frac{\varpi'}{\varpi} f' = -\frac{1}{\varpi}, \qquad f(0)=1, \qquad f'(0)=1,
\]
whose solution can be written as 
\[
f = 1 - \int_0^{\bar{\tau}} \left(\frac{1}{\varpi^2} \int_0^{s} \varpi \,\mbox{d}s' \right) \mbox{d}s.
\]
From the latter one obtains the chain of inequalities
\begin{eqnarray*}
&& 0 \leq 1-\frac{\varsigma}{\varpi} = 1-f= 
\int_0^{\bar{\tau}} \left(\frac{1}{\varpi^2} \int_0^{s} \varpi \,\mbox{d}s' \right) \mbox{d}s \\
&& \hspace{5cm}\leq \int_0^{\bar{\tau}} \frac{s}{\varpi}\, \mbox{d} s \leq \int_0^{\bar{\tau}} \frac{s}{\cosh (\beta s)} \mbox{d}s \leq 2 \int_0^{\bar{\tau}} s  e^{-2\beta s} \mbox{d} s.
\end{eqnarray*}
Hence, one has that
\[
0 \leq 1-\frac{\varsigma}{\varpi} \leq \frac{2}{\beta^2}\left(1-(\beta \bar{\tau}+1) e^{-\beta \bar{\tau}}\right) \leq \frac{2}{\beta^2}.
\]
For the values of $r_*>r_\circledast$ for which $\zeta<0$ a
direct computation shows that $-1 <2\zeta/\beta^2< 0$.  Thus,
\[
\omega_* + \left(1 -\frac{\varsigma}{\varpi}\right) \zeta \geq
\omega_* + \frac{2\zeta}{\beta^2} > \omega_* -1 .
\]
For the values of $r_*>r_\circledast$ for which $\zeta>0$  one readily has that 
\[
\omega_* + \left(1 -\frac{\varsigma}{\varpi}\right) \zeta \geq
\omega_*.
\] 
Hence, in both cases using equation \eqref{DeviationData1} one concludes that
\[
\omega_* + \left(1 -\frac{v}{u}\right) \zeta \geq \omega_*-1 =
\frac{1}{\varrho_*}(m^2-q^2 +4\varrho_* m)> 0.
\]

In order to conclude the argument we consider $\Theta \omega$ where
$\Theta$ is given by equation
\eqref{ConformalFactorReducedPhysical}. Putting together the
discussion from the previous paragraphs one has that
\begin{eqnarray*}
&& \Theta \omega = \Theta \varpi \left( \omega_* + \left(1 -\frac{\varsigma}{\varpi}\right)
  \zeta \right) \geq \frac{\Theta \varpi}{\varrho_*}(m^2-q^2 +4\varrho_* m)
\\
&& \phantom{\Theta \omega = \Theta \varpi \left( \omega_* + \left(1 -\frac{\varsigma}{\varpi}\right)
  \zeta \right)} \geq \Theta_*
\frac{\cosh(\bar{\tau})}{\varrho_*\cosh^2(\tfrac{1}{2}\beta\bar{\tau})}(m^2-q^2
+4\varrho_* m) \\
&&\phantom{\Theta \omega = \Theta \varpi \left( \omega_* + \left(1 -\frac{\varsigma}{\varpi}\right)
  \zeta \right)} \geq \frac{\Theta_*}{\varrho_*}(m^2-q^2 +4\varrho_* m)>0.
\end{eqnarray*}
This lower bound holds even in the limit $\bar{\tau}\rightarrow
\infty$ (i.e. $\tau=0$). By the same considerations made in Section
\ref{Section:DeviationCriticalCurve}, it follows that the congruence of
conformal curves remains free of conjugate points even at null
infinity.

\subsection{Curves with  $r_*<r_\circledast$}
The observation made in Section \ref{Section:CGInside} that for
$r_*<r_\circledast$, curves with a certain value of $r_*$ can have a
lower value of the turning point $\bar{r}=\alpha_2$ than curves
starting closer to the horizon, shows that curves in our congruence of
conformal curves must intersect at some point. This is because curves
with constant coordinate value $r$ are timelike, respectively, in the
regions III of the non-extremal case and the regions II of the
extremal case. The question is then: how long does this part of the
congruence exist without the existence of conjugate points? In the
sequel it will be shown that the congruence of conformal curves is
free of conjugate points up, and including, the horizon. More precisely one has that:

\begin{proposition}
\label{Proposition:NonExistenceConjugatePoints}
Assume that $q^2\leq \tfrac{8}{9}m^2$ and $r_*\in (r_+,r_\circledast)$. Then, the solutions to equation
\eqref{enemy} with initial data given by
\eqref{DeviationData1}-\eqref{DeviationData2} satisfy $\omega>0$ for $\bar{\tau}\in[0,\bar{\tau}_{\mathscr{H}^+}]$.
\end{proposition}

The case $\tfrac{8}{9}m^2<q^2<m^2$ will not be considered
here. Instead, we concentrate our attention on the extremal case for
which one can prove the following:

\begin{proposition}
\label{Proposition:NonExistenceConjugatePointsExtremal}
Assume $q^2=m^2$. Then there exists $r_\star\in (r_+,r_\circledast)$
such that for $r_*\in (r_\star,r_\circledast)$,   the solutions to equation
\eqref{enemy} with initial data given by
\eqref{DeviationData1}-\eqref{DeviationData2} satisfy $\omega>0$ for $\bar{\tau}\in[0,\bar{\tau}_{\mathscr{H}^+}]$.
\end{proposition}

The proof of this propositions is given in Sections
\ref{Section:SimpleCase} and \ref{Section:DifficultCase},
respectively. Clearly, the result of Proposition
\ref{Proposition:NonExistenceConjugatePointsExtremal} (which includes
the extremal case) is much more restrictive than that of Proposition
\ref{Proposition:NonExistenceConjugatePoints}. In view of the discussion of Section \ref{Section:TimelikeInfinity},
the extremal Reissner-Nordstr\"om spacetime is the case of most
relevance from the perspective of conformal geometry. In this respect,
the result given in Proposition
\ref{Proposition:NonExistenceConjugatePoints} will be sufficient for
future applications to be considered elsewhere. Numerical evaluation
of the solutions of equation \eqref{enemy} suggest, nevertheless, that the conclusions of Proposition
\ref{Proposition:NonExistenceConjugatePoints} can be extended to the
whole range $q^2\leq m^2$ and $r_*\in (r_+,r_\circledast)$ ---so that,
in particular, Proposition
\ref{Proposition:NonExistenceConjugatePointsExtremal} could be
superseded. This proof would require an analysis which would increase
considerably the length of this article. 

\medskip
 For simplicity of the
presentation, \emph{in the remainder of the subsection it is always
assumed that $\bar{r}\in [\alpha_2,r_*]$. For these values of
$\bar{r}$ one has that $\bar{r}'<0$ with $\bar{r}'=0$ only at
$\bar{r}=r_*,\,\alpha_2$.} The amount of technical details in the analysis of the solutions to equation
\eqref{enemy} depends on whether the value of the charge, $q$, is close or
not to the extremal value. In order to characterise the values of $q$
requiring a more careful treatment, it is recalled that $\bar{\chi}\geq 0$ if
$\bar{r}\geq 3q^2/2m$.  As $\bar{r}$ is monotonically decreasing, it follows that
for $\bar{r}\in[r_*,r_+]$, $\bar{\chi}\geq0$ if and only if 
\[
\frac{3q^2}{2m} \leq r_+=m+ \sqrt{m^2-q^2}.
\]
The above inequality is saturated if $q^2=\tfrac{8}{9}m^2$. Our
subsequent discussion is split depending on whether $q^2$ is below or
above the critical value found in the previous lines.

\subsubsection{Proof of Proposition \ref{Proposition:NonExistenceConjugatePoints}}
\label{Section:SimpleCase}

As already discussed, if $q^2\leq \tfrac{8}{9}m^2$, one has that
$\bar{\chi}\geq 0$ for $\bar{r}\in[r_+,r_*]$. It follows then by an argument
similar to the one used in Subsection
\ref{Subsection:DeviationOutside} that $\omega>0$ for $\bar{\tau}\in
[0,\bar{\tau}_{\mathscr{H}^+}]$. The full details will not be
provided, but it is noticed that, in fact, the task in this case is
simpler as one is dealing with finite values of $\bar{\tau}$ and
$\bar{r}$. Hence, it is only necessary to ensure the positivity of
$\omega$ and not that of $\Theta\omega$.

\subsubsection{Proof of Proposition \ref{Proposition:NonExistenceConjugatePointsExtremal}}
\label{Section:DifficultCase}

All through out it is assumed that $q^2=m^2$. In this case, the analysis of the previous sections show that there
are intervals of $\bar{\tau}$ for which $\bar{\chi}<0$, so that the
arguments of Subsection \ref{Subsection:DeviationOutside}, do not
apply for the whole interval $ [0,\bar{\tau}_{\mathscr{H}^+}]$. Hence,
a more detailed analysis is required. In particular, formula
\eqref{FormulaPsi} shows that in the extremal case $\bar{\chi}$ is always negative
for $\bar{r}\in[\alpha_2,r_*]$. As before, we will restrict our
attention to the behaviour of the conformal curves in the range
$\bar{\tau}\in [0,\bar{\tau}_{\mathscr{H}^+}]$. In such interval
$\bar{r}(\bar{\tau})$ is a monotonic decreasing function of
$\bar{\tau}$ with $\bar{r}(0)=r_*$ and
$\bar{r}(\bar{\tau}_{\mathscr{H}^+})=r_+$. Hence, for $\bar{\tau}\in
(0,\bar{\tau}_{\mathscr{H}^+}]$ it is convenient to reparametrise
equation \eqref{enemy} in terms of $\bar{r}$. Using the chain rule to
write
\[
\omega' = \bar{r}'\frac{\mbox{d}\omega}{\mbox{d}\bar{r}}, \qquad
\omega''=\bar{r}'\frac{\mbox{d}}{\mbox{d}\bar{r}}\left( \bar{r}'\frac{\mbox{d}\omega}{\mbox{d}\bar{r}} \right)
\]
one readily has that the deviation equation \eqref{enemy} implies the equation
\begin{equation}
\bar{r}^{\prime 2}\stackrel{\backprime\backprime}{\omega} + \bar{r}''
\stackrel{\backprime}{\omega} -(\beta^2 + \bar{\chi}) \omega =\zeta,
\label{Equationw:r}
\end{equation}
with
\[
\stackrel{\backprime}{\omega} \equiv
\frac{\mbox{d}\omega}{\mbox{d}\bar{r}}, \qquad
\stackrel{\backprime\backprime}{\omega}\equiv \frac{\mbox{d}^2 \omega}{\mbox{d}\bar{r}^2}
\]
and where $\bar{\chi}$ is now regarded as a function of $\bar{r}$, and
$\bar{r}^{\prime 2}$ is given by equation
\eqref{ReducedSqrtRoots}. Notice that as $\bar{r}'=0$ at
$\bar{r}=r_*,\,\alpha_2$ equation \eqref{Equationw:r} is formally
singular. An explicit formula for $\bar{r}''$ in terms of $\bar{r}$ can be found using
\[
\bar{r}'' = \frac{1}{2} \frac{1}{\bar{r'}} \frac{\mbox{d}
}{\mbox{d}\bar{\tau}}\left(\bar{r}^{\prime2}\right) = \frac{1}{2}
\frac{\mbox{d}}{\mbox{d}\bar{r}}\left( r^{\prime2}\right). 
\]
One obtains
\begin{eqnarray}
&& \bar{r}'' = \frac{\beta^2}{2\bar{r}^3} \big(
(\bar{r}-\alpha_1)(\bar{r}-\alpha_2)(\bar{r}-\alpha_3)(2\bar{r}_*-\bar{r})
+ \bar{r}(\bar{r}-\bar{r}_*)(\bar{r}-\alpha_2)(\bar{r}-\alpha_3)
\nonumber \\
&&
\hspace{3cm}+\bar{r}(\bar{r}-\bar{r}_*)(\bar{r}-\alpha_1)(\bar{r}-\alpha_3)+\bar{r}(\bar{r}-\bar{r}_*)(\bar{r}-\alpha_1)(\bar{r}-\alpha_2)
\big). \label{rprimeprimeExpanded}
\end{eqnarray}
In particular, using the information from Lemma
\ref{Lemma:BehaviourRoots} one can readily conclude that
\[
\bar{r}''_*<0, \qquad \bar{r}''_2>0
\]
where $\bar{r}''_*\equiv \bar{r}''(r_*)$ and $\bar{r}''_2\equiv
\bar{r}''(\alpha_2)$. Thus, one concludes that there exists $r_!\in
(\alpha_2,r_\circledast)$ such that $\bar{r}''_!\equiv
\bar{r}''(r_!)=0$. It can be verified that this zero of
$\bar{r}''$ in $(\alpha_2,r_\circledast)$ is unique. An analysis of
formula \eqref{rprimeprimeExpanded}  yields the bounds
\[
m< r_! <\tfrac{11}{10}m.
\]

\medskip
The initial data for equation \eqref{Equationw:r} is given by
\begin{subequations}
\begin{eqnarray}
&& \omega(r_*)= \omega_*\label{InitialEquationw:rA}\\
&& \stackrel{\backprime}{\omega}_*\equiv \stackrel{\backprime}{\omega}(r_*) =
\lim_{\bar{\tau}\rightarrow 0} \frac{\omega'}{\bar{r}'} =
\lim_{\bar{\tau}\rightarrow 0} \frac{\omega''}{\bar{r}''} = \frac{\omega''_*}{\bar{r}''_*} , \label{InitialEquationw:rB}
\end{eqnarray}
\end{subequations}
where
\[
\omega''_* = \zeta + (\beta^2 + \chi_*)\omega_*, \qquad \bar{r}''_* = \frac{\beta^2}{2\bar{r}_*^2}(r_*-\alpha_1)(r_*-\alpha_2)(r_*-\alpha_3).
\]
Using that $\bar{r}''_*<0$, and noticing that $\omega''_*>0$ if $r_*\geq
\sqrt{\tfrac{3}{2}}m$ one concludes that
\[
\stackrel{\backprime}{\omega}_*<0, \qquad \mbox{if} \qquad 
r_*\in ( \sqrt{\tfrac{3}{2}}m,r_\circledast). 
\]
Alternatively, one can compute $\stackrel{\backprime}{\omega}_*$
directly by evaluating equation \eqref{Equationw:r} on
$r_*$. Similarly, differentiating \eqref{Equationw:r} with respect to
$\bar{r}$ and evaluating on $r_*$ one finds that
\[
\stackrel{\backprime\backprime}{\omega}_*<0, \qquad r_*\in
(r_+,r_\circledast).
\]
In particular, in the extremal case one has the following expressions:
\begin{equation}
\stackrel{\backprime}{\omega}_* = \frac{(3m^2
  -2r_*^2)}{(3m-2r_*)(m-r_*)^2}, \qquad
\stackrel{\backprime\backprime}{\omega}_* = \frac{2m(r_*-2m)}{r_*(3m-2r_*)(m-r_*)^2}. 
\label{InitialDDomega}
\end{equation}
Notice, in particular, that both
$\stackrel{\backprime}{\omega}_*,\,\stackrel{\backprime\backprime}{\omega}_*\rightarrow
\infty$ as $r_*\rightarrow r_\circledast=\tfrac{3}{2}m$. 

\medskip
\noindent
\textbf{Analysis of $\stackrel{\backprime\backprime}{\omega}$.} As it
will be seen in the sequel, the proof of Proposition
\ref{Proposition:NonExistenceConjugatePointsExtremal} requires a
knowledge of the sign of $\stackrel{\backprime\backprime}{\omega}$. In
order to analyse this, it is convenient to consider a \emph{first
integral} of equation \eqref{Equationw:r}. Multiplying
\eqref{Equationw:r} by $\stackrel{\backprime}{\omega}$ one readily
obtains that
\[
\tfrac{1}{2}\bar{r}^{\prime 2} ( \stackrel{\backprime}{\omega}{}^2)^\backprime + \bar{r}'' \stackrel{\backprime}{\omega}{}^2 -\tfrac{1}{2} (\beta^2 + \bar{\chi}) (\omega^2)^\backprime = \zeta \stackrel{\backprime}{\omega}.  
\]
Integrating with respect to $\bar{r}\in[r_+,r_*]$ leads to 
\[
 \tfrac{1}{2}\int^{r_*}_{\bar{r}}\bar{r}^{\prime 2} ( \stackrel{\backprime}{\omega}{}^2)^\backprime \mbox{d}s + \int^{r_*}_{\bar{r}}\bar{r}'' \stackrel{\backprime}{\omega}{}^2 \mbox{d}s -\tfrac{1}{2} \int^{r_*}_{\bar{r}}(\beta^2 + \bar{\chi}) (\omega^2)^\backprime  \mbox{d}s= \zeta \int^{r_*}_{\bar{r}}\stackrel{\backprime}{\omega} \mbox{d}s.
\]
Integration by parts in the first and third terms yields
\[
\tfrac{1}{2} \bar{r}^{\prime 2} \stackrel{\backprime}{\omega}{}^2 \bigg|^{r_*}_{\bar{r}} - \tfrac{1}{2} \int_{\bar{r}}^{r_*} (\bar{r}^{\prime 2})^\backprime \stackrel{\backprime}{\omega}{}^2 \mbox{d}s + \int_{\bar{r}}^{r_*} \bar{r}'' \stackrel{\backprime}{\omega}{}^2 \mbox{d}s - \tfrac{1}{2} (\beta^2 + \bar{\chi})\omega^2 \bigg|_{\bar{r}}^{r_*} + \tfrac{1}{2} \int_{\bar{r}}^{r_*} \stackrel{\backprime}{\bar{\chi}} \omega^2 \mbox{d}s = \zeta (\omega_* -\omega).
\]
Recalling that $\bar{r}^{\prime 2}_*=0$ and that $\bar{r}'' = \tfrac{1}{2} (\bar{r}^{\prime 2})^\backprime$, this last expression reduces to
\begin{equation}
-\tfrac{1}{2} \bar{r}^{\prime 2} \stackrel{\backprime}{\omega}{}^2 + \tfrac{1}{2} (\beta^2 + \bar{\chi}) \omega^2 -\tfrac{1}{2}(\beta^2 + \bar{\chi}_*) \omega^2_* + \tfrac{1}{2} \int^{r_*}_{\bar{r}} \stackrel{\backprime}{\bar{\chi}} \omega^2 \mbox{d}s = \zeta(\omega_*-\omega). 
\label{FirstIntegral}
\end{equation}
This equation will be used, in the sequel to prove the following result:

\begin{lemma}
\label{Lemma:SignDDomega}
The solution, $\omega$, of equation \eqref{Equationw:r} with initial
data given by \eqref{InitialEquationw:rA} and
\eqref{InitialEquationw:rB}, $r_*\in(r_+,r_\circledast)$ satisfies
$\stackrel{\backprime\backprime}{\omega}<\stackrel{\backprime\backprime}{\omega}_*$ for $r\in
(\alpha_2,r_*)$.
\end{lemma}

\begin{proof}
The proof proceeds by contradiction. Hence, assume that there exists
$\bar{r}=r_\lightning$ such that $  r_* \neq r_\lightning \neq r_!$ and
$\stackrel{\backprime\backprime}{\omega}_\lightning \equiv
\stackrel{\backprime\backprime}{\omega}(r_\lightning)
=\stackrel{\backprime\backprime}{\omega}_*$. Equation
\eqref{Equationw:r} implies that
\[
(\stackrel{\backprime}{\omega}_\lightning)^2 =
\frac{1}{(\bar{r}''_\lightning)^2} \left(\zeta + (\beta^2
  +\chi_\lightning)\omega_\lightning-\bar{r}^{\prime 2}_\lightning\stackrel{\backprime\backprime}{\omega}_*\right)^2.
\] 
Substituting this expression for
$(\stackrel{\backprime}{\omega}_\lightning)^2$ into the first integral
\eqref{FirstIntegral} with $\bar{r}=r_\lightning$ and grouping terms
one obtains
\begin{equation}
a_2 \omega^2_\lightning + a_1 \omega_\lightning + a_0 = (\bar{r}''_\lightning)^2\int_{r_\lightning}^{r_*} \stackrel{\backprime}{\bar{\chi}} \omega^2 \mbox{d}s, 
\label{IntermediateStep}
\end{equation}
with
\begin{eqnarray*}
&& a_2 \equiv (\beta^2+\bar{\chi}_\lightning)\left( (\beta^2+\bar{\chi}_\lightning)\bar{r}^{\prime 2}_\lightning - (\bar{r}'' _\lightning)^2\right),\\
&& a_1 \equiv \left( 2(\beta^2+\bar{\chi}_\lightning)\bar{r}_\lightning^{\prime 4}\stackrel{\backprime\backprime}{\omega}_*+2\zeta(\beta^2 +\bar{\chi}_\lightning)\bar{r}^{\prime2}_\lightning-2\zeta (\bar{r}'' _\lightning)^2\right),\\
&& a_0 \equiv \left((\bar{r}^{\prime 2}_\lightning-2\zeta) \bar{r}^{\prime 4}_\lightning\stackrel{\backprime\backprime}{\omega}_\lightning+\bar{r}^{\prime 2}_\lightning \zeta^2 + 2 \zeta (\bar{r}'')^2 _\lightning\omega_* + (\bar{r}'')^2 _\lightning (\beta^2 + \bar{\chi}_*)\omega_*^2\right).
\end{eqnarray*}
The coefficients $a_0,\, a_1,\, a_0$ are explicitly
known rational expressions of $r_\lightning$. Now, as
$\stackrel{\backprime}{\bar{\chi}}>0$ for $\bar{r}\in [\alpha_2,r_*]$, it
follows from equation \eqref{IntermediateStep} that
\begin{equation}
a_2 \omega^2_\lightning + a_1 \omega_\lightning + a_0 >0, \qquad
r_\lightning\neq r_!.
\label{IntermediateStep2}
\end{equation}
A lengthy computation using the explicit expressions for $a_2$, $a_1$ and $a_0$ in terms of $\bar{r}$ shows that $a_2<0$ and $a_1^2-4
a_2a_0<0$ for $r_\lightning\in (\alpha_2,r_*)$ so that the polynomial
$f(x) =a_2 x^2 + a_1 x +a_0$ is negative for $x\in \mathbb{R}$. This
is a contradiction with \eqref{IntermediateStep2}. Thus, assuming that
$r_\lightning\neq r_*,\; r_!$, there is no value of $\bar{r}$ for which
$\stackrel{\backprime\backprime}{\omega}\geq
\stackrel{\backprime\backprime}{\omega}_*$. The possibility $r_\lightning=
r_!$ can be excluded by continuity.
\end{proof}

In what follows, we restrict our attention to curves such that
$r_*\in(\sqrt{\frac{3}{2}}m,r_\circledast)$ so that
$\stackrel{\backprime}{\omega}_*<0$ ---cf. equation
\eqref{InitialDDomega}. Consistent with Lemma  \ref{Lemma:SignDDomega}
it is assumed that  $\omega$ has a local maximum in $(r_+,r_*)$
---otherwise one directly has that $\omega_+>0$ and there is nothing
to prove. Denote by
$r_\wedge$ the location of such local extremum and write
$\omega_\wedge\equiv \omega(r_\wedge)$. A lower bound for
$\omega_\wedge$ can be obtained from the evaluation of equation
\eqref{Equationw:r} at $\bar{r}=r_!$. As $\bar{r}''_! =0$, using Lemma
\ref{Lemma:SignDDomega} one concludes
that
\[
\bar{r}^{\prime 2}_! \stackrel{\backprime\backprime}{\omega}_!= \zeta +(\beta^2 +\chi_!) \omega_! \leq
\bar{r}^{\prime 2}_! \stackrel{\backprime\backprime}{\omega}_*<0.
\]
It follows then that 
\begin{equation}
\frac{\bar{r}^{\prime 2}_!
  \stackrel{\backprime\backprime}{\omega}_*-\zeta}{\beta^2+\chi_!}
\leq \omega_!\leq \omega_\wedge.
\label{BoundMaximum:omega}
\end{equation}
It can be readily verified that $\beta^2 + \chi_!$ is negative so that
the lower bound of $\omega_\wedge$ given by the inequality
\eqref{BoundMaximum:omega} is positive. Important for the sequel is
the following observation: combining inequality
\eqref{BoundMaximum:omega} with the value for
$\stackrel{\backprime\backprime}{\omega}_*$ given in \eqref{InitialDDomega} it
follows that
\[
\omega_\wedge \rightarrow \infty \quad \mbox{ as } \quad
r_*\rightarrow r_\circledast.
\]

\medskip
The subsequent analysis will also require
of an upper bound for $r_\wedge$. Such a bound can be more easily
obtained by considering equation \eqref{enemy}. As already discussed,
for $r_*\in(\sqrt{\frac{3}{2}}m,r_\circledast)$ one has that
$\omega''_*>0$. As $\omega$ is assumed to have a maximum, it follows that there must
exist an inflexion point at which $\omega''=0$. At this inflexion
point equation \eqref{enemy} implies $-(\beta^2 + \bar{\chi})\omega
=\zeta$. This last equality, together with the observation that at this point
$\omega\geq \omega_*$, leads to
\begin{equation}
-\frac{\zeta}{\beta^2+\bar{\chi}}\geq \omega_*.
\label{InequalityLocationInflexion}
\end{equation}
The above inequality can be regarded as a condition on
$\bar{r}$ as $\zeta>0$ for the range of $r_*$ under consideration
---cf. \eqref{ZetaExtremal}. Notice, in particular, that $\beta^2+\bar{\chi}>0$ in
order for \eqref{InequalityLocationInflexion} to make sense as
$\omega_*> 1$. Some inspection shows that for $r_*\in(\sqrt{\frac{3}{2}}m,r_\circledast)$
\[
r_\wedge < r_\dagger\equiv \tfrac{13}{10}m. 
\]
A final observation concerning the maximum of $\omega$ is the
following: evaluating equation \eqref{Equationw:r} at the maximum one
obtains after some rearranging that
\[
\frac{\bar{r}^{\prime 2}_\wedge \stackrel{\backprime\backprime}{\omega}_\wedge
    -\zeta}{\beta^2 + \bar{\chi}_\wedge}=\omega_\wedge>0.
\]
However, $\bar{r}^{\prime 2}_\wedge \stackrel{\backprime\backprime}{\omega}_\wedge
    -\zeta <0$ so that necessarily
    $\beta^2+\bar{\chi}_\wedge <0$. 

\medskip
\noindent
\textbf{Estimating $\omega_+$.} In what follows, assume that
$r_*\in(\sqrt{\frac{3}{2}}m,r_\circledast)$. Moreover, suppose 
that $r_+<r_\wedge$ ---otherwise, as a consequence of Lemma
\ref{Lemma:SignDDomega} one has that $\omega_+\neq 0$ and the result
of Proposition \ref{Proposition:NonExistenceConjugatePointsExtremal} follows directly. In order to estimate
$\omega_+$ we exploit the information on the location and the size of
the maximum of $\omega$ acquired in the previous section. 

Clearly, $\omega>0$ for $\bar{r}$ sufficiently close to $r_\wedge$. We
make use of a bootstrap argument to show that the interval where
$\omega>0$ can be extended to include $r_+$. A
calculation similar to the one leading to equation
\eqref{FirstIntegral} yields
\begin{equation}
\bar{r}^{\prime 2} \stackrel{\backprime}{\omega}{}^2
=2\zeta(\omega-\omega_\wedge)+ (\beta^2 + \bar{\chi}) \omega^2
-(\beta^2 + \bar{\chi}_\wedge) \omega^2_\wedge +
\int^{r_\wedge}_{\bar{r}} \stackrel{\backprime}{\bar{\chi}} \omega^2
\mbox{d}s, \qquad \bar{r}\in [r_+,r_\wedge]
\label{FirstIntegralFromMaximum}
\end{equation}
where, in particular, it has been used that
$\stackrel{\backprime}{\omega}_\wedge=0$. One needs to estimate 
various terms in this expression. To this end it is noticed that $r_+ \leq \bar{r}\leq r_\wedge < r_\dagger$,
so that using formula
\eqref{ReducedSqrtRoots} it follows that
\begin{equation}
C^2_1
\stackrel{\backprime}{\omega}{}^2 \leq \bar{r}^{\prime 2}
\stackrel{\backprime}{\omega}{}^2\qquad \mbox{for} \qquad \bar{r}\in[r_+,r_\wedge],
\label{Estimate1}
\end{equation}
where
\[
C^2_1 \equiv \frac{\beta^2}{r_\dagger^2}
(r_+-\alpha_1)(r_+-\alpha_2)(r_*-r_\dagger)(\alpha_3-r_\dagger)>0.
\]
Moreover, it can be explicitly verified that
$\stackrel{\backprime\backprime}{\bar{\chi}}<0$ so that
\begin{equation}
\int^{r_\wedge}_{\bar{r}} \stackrel{\backprime}{\bar{\chi}} \omega^2
\mbox{d}s < \;\stackrel{\backprime}{\bar{\chi}}_+
\omega^2_\wedge(r_\wedge-r_+) \qquad \mbox{as long as} \qquad \omega> 0. 
\label{Estimate2}
\end{equation}
Finally, one has that
\begin{equation}
\beta^2 + \bar{\chi}<0 \qquad  \mbox{for} \qquad \bar{r}\in [r_+,r_\wedge].
\label{Estimate3}
\end{equation}
Making use of inequalities \eqref{Estimate1}-\eqref{Estimate3} in equation \eqref{FirstIntegralFromMaximum} one
concludes that
\[
C^2_1 \stackrel{\backprime}{\omega}{}^2 <
2\zeta(\omega-\omega_\wedge) -(\beta^2 + \bar{\chi}_\wedge)
\omega^2_\wedge + \stackrel{\backprime}{\bar{\chi}}_+
\omega^2_\wedge(r_\wedge-r_+), \quad \mbox{on} \quad
[\bar{r},r_\wedge] \quad \mbox{as long as} \quad \omega>0.
\]
For the convenience of the presentation let 
\[
C_2 \equiv  \stackrel{\backprime}{\bar{\chi}}_+ \omega^2_\wedge(r_\wedge-r_+) -(\beta^2 + \bar{\chi}_\wedge)
\omega^2_\wedge-2\zeta \omega_\wedge,
\]
so that 
\begin{equation}
0<C_1^2 \stackrel{\backprime}{\omega}{}^2 \leq 2\zeta \omega + C_2.
\label{BasicInequality}
\end{equation}
As $\stackrel{\backprime}{\omega}>0$ on $[r_+,r_\wedge]$,
one can consider the positive square root of inequality
\eqref{BasicInequality} and then integrate over $[\bar{r},r_\wedge]$ to obtain
\begin{equation}
\sqrt{2\zeta \omega_\wedge + C_2} +\frac{\zeta}{C_1}(\bar{r}-r_\wedge)
<  \sqrt{2\zeta \omega +C_2}.
\label{BasicBoundOmega}
\end{equation}
The second term of the left hand side of this last inequality is
negative as $\bar{r}<r_\wedge$. However, in view of the second equation in
\eqref{InitialDDomega}
and the bound \eqref{BoundMaximum:omega} it is possible to ensure that the left
hand side is positive by choosing $r_*$ sufficiently close to
$r_\circledast$ ---that is, there exists $r_\star \in
(\sqrt{\frac{3}{2}}m,r_\circledast)$ such that if $r_*\in
(r_\star,r_\circledast)$ then
\[
0<\sqrt{2\zeta \omega_\wedge + C_2} +\frac{\zeta}{C_1}(r_+-r_\wedge)
\leq \sqrt{2\zeta \omega_\wedge + C_2} +\frac{\zeta}{C_1}(\bar{r}-r_\wedge).
\]
Crucially, it can be verified that $C_1$ remains finite
and non-zero as $r_*\rightarrow r_\circledast$.  Squaring inequality \eqref{BasicBoundOmega}
and simplifying one obtains the lower bound
\[
\omega_\wedge + \frac{\zeta}{2C_1^2}(r_+-r_\wedge)^2 + \frac{1}{C_1}
(r_+-r_\wedge) \sqrt{2\zeta \omega_\wedge + C_2} < \omega.
\]
Using the definition of $C_2$, this inequality can be rewritten
as
\begin{equation}
C_3 \omega_\wedge + \frac{\zeta}{2C_1^2}(r_+-r_\wedge)^2 < \omega, 
\label{FinalInequality}
\end{equation}
where
\[
C_3 \equiv 1+ \frac{1}{C_1}(r_+-r_\wedge) \sqrt{\stackrel{\backprime}{\bar{\chi}}_+(r_\wedge-r_+)-(\beta^2+\bar{\chi}_\wedge)}.
\]
A lengthy direct computation using the information available about the various terms in
this expression shows that $C_3>0$ for $r_*\in (r_\star,r_\circledast)$. As $\zeta>0$ for the range of $r_*$ under
consideration it follows from inequality \eqref{FinalInequality} that
$\omega>C_4>0$ where $C_4$ is independent of
$\bar{r}\in[r_+,r_\circledast]$ ---at least for curves with $r_*$ close enough to
$r_\circledast$. In particular, one has that $\omega_+>0$. This proves Proposition
\ref{Proposition:NonExistenceConjugatePointsExtremal}.

\subsection{Some remarks}

In the present
discussion, no attempt has been made to analyse the behaviour of the
congruence after it crosses the horizon. However, numerical
evaluations of equation \eqref{enemy} show that the scalar $\omega$
goes to zero shortly after the curve has crossed the horizon, and
certainly, before it reaches the turning point $\bar{r}=\alpha_2$. In
any case, one knows there exists an open neighbourhood after the
horizon where the congruence remains non-degenerate.

\section{Conclusions}
\label{Section:Conclusions}

The analysis carried out in Sections
\ref{Section:ExplicitExpressions}, \ref{Section:AnalysisConformalCurves}
and \ref{Section:ConformalDeviationEquations} allows to provide the following technical version of our main Theorem:

\begin{theorem}
\label{MainTheoremTechnicalVersion}
Let $(\tilde{\mathcal{M}},\tilde{\bm g})$ denote a Reissner-Nordstr\"om
spacetime with $q^2\leq m^2$ and let $r_\circledast$ as defined by equation
\eqref{rCircledast}. On $(\tilde{\mathcal{M}},\tilde{\bm g})$ consider
the congruence of timelike conformal curves defined by the initial conditions
\eqref{CGInitialData} and $r_*\in(r_+,\infty)$ on the time symmetric slice of the domain of
outer communication. Let $\bar{\tau}$ and $\tau$ denote,
respectively, the physical and conformal proper
time of the curves of the congruence. For this congruence one has that:

\begin{itemize}

\item[(a)]  Each curve of this congruence
exists for $\bar{\tau}\in \mathbb{R}$. Moreover:
\begin{itemize}
\item[(i)] the curves with $r_*\in(r_\circledast,\infty)$ reach null infinity in
  an infinite amount of physical proper time but in a finite amount of conformal proper time; 
\item[(ii)] the curves with $r_*=r_\circledast$ reach past and future
  timelike infinity in an infinite amount of physical proper time but a finite amount of conformal proper time;
\item[(iii)] the curves  with $r_*\in(r_+,r_\circledast)$ reach the
  event horizon in a finite amount of both physical and conformal
  proper time.
\end{itemize}

\item[(b)] In addition one has that:

\begin{itemize}
\item [(i)] If $q^2 \leq \tfrac{8}{9}m^2$ then the congruence if free
  of conjugate points in the domain of outer communication.

\item[(ii)] In the extremal case $q^2=m^2$, there exists
$r_\star\in(r_+,r_\circledast)$ such that the subcongruence defined by
$r_*\in (r_\star,\infty)$ is free of conjugate points in the domain of
outer communication.

\end{itemize}

\end{itemize}

\end{theorem}

As already indicated in the main text, numerical evaluations of the
congruence suggest that it should be possible to improve Theorem
\ref{MainTheoremTechnicalVersion} so as to ensure that the congruence of conformal
curves is free of conjugate points in the domain of outer
communication for $q^2\leq m^2$. 

\medskip
The analysis of this article a first step in the study of the
Reissner-Nordstr\"om spacetime as a solution of the conformal field
equations. In
view of this programme, the results of Section
\ref{Section:TimelikeInfinity} are specially relevant as they suggests
that the conformal structure of the timelike infinity, $i^+$, of the
extremal Reissner-Nordstr\"om spacetime may be more tractable, from an
analytic point of view, than that of the non-extremal case.

Regarding the Reissner-Nordstr\"om spacetime as a spherically
symmetric model of the Kerr spacetime, it is natural to wonder how
much of the structure observed in the present analysis has an analogue
in the Kerr solution. For example, it is natural to conjecture that
the domain of outer communication of the Kerr spacetime can be covered
by means of a non-singular congruence of conformal
geodesics reaching beyond null infinity. It is very likely that this
congruence will degenerate after it has crossed the event horizon and
that the curves will have some type of singularity avoiding properties
so that there may exists regions in the black hole region which can not
be probed in this way. A more tantalising possibility is that, as in
the case of the extremal Reissner-Nordstr\"om spacetime, the extreme
Kerr may have a more tractable structure at $i^+$. In any case, the
analysis of conformal geodesics in the Kerr spacetime is bound to be
much more complicated as  the warped product structure of the
line element is lost.

\section*{Acknowledgements}
We have profited from discussions with H. Friedrich, S. Dain and
J.M. Heinzle. We also thank A. Garcia-Parrado, D. Pugliese and
A. Carrasco for useful comments and observations. CL acknowledges
funding from the project grant FCT/CERN/FP/123609/2011. JAVK thanks
the hospitality of the Erwin Schr\"odinger Institute for Mathematical
Physics in Vienna during the programme "Dynamics of General
Relativity: black holes and asymptotics" in December 2012 when the
last stages of this research were carried out.

\appendix

\section{Conformal geodesics in the Schwarzschild spacetime}
\label{Appendix:Schwarzschild}

For completeness, we include a study of the
solutions to the conformal curve equations in the case of the
Schwarzschild spacetime (where $q=0$). In this case the conformal
curves are, in fact, conformal geodesics. The analysis of these curves
was originally done in \cite{Fri03c} using explicit solutions in terms
of elliptic functions. The discussion given here follows the strategy
of section \ref{Section:AnalysisConformalCurves} in the main text, and
avoids the use of explicit solutions.

\medskip
As in the case of the main text, essential for our analysis is the
factorisation of the polynomial appearing in equation
\eqref{ReducedSqrt}. If $q=0$,  then $P(\bar{r})$ is of degree $3$ and
one has the factorisation
\begin{equation}
P(\bar{r}) =\beta^2 (\bar{r}-r_*)(\bar{r}-\alpha)(\bar{r}+\alpha),
\label{FactorisationSchwarzschild}
\end{equation}
where 
\[
\alpha \equiv \sqrt{\frac{m r_*}{2D_*}}.
\]
The constant solution to equation \eqref{ReducedSqrt} can be found to
be given by $\bar{r}=r_\circledast=\tfrac{5}{2}m$. As discussed in
\cite{Fri03c}, these curves reach the point $i^+$ in a finite amount
of unphysical proper time, and divide the two possible regimes for the
conformal geodesics. 

\subsection{Conformal geodesics with $r_*>r_\circledast$}
If $r_\circledast<r_*$, then the analysis of the curves is covered by
the discussion in section \ref{Section:CGOutside} by setting
$q=0$. These conformal geodesics reach null infinity.

\subsection{Conformal geodesics with $r_*<r_\circledast$}
If $r_*<r_\circledast$, then one can readily verify that
$\bar{r}''_*<0$, so that $\bar{\tau}=0$ is a maximum of the function
$\bar{r}$ as one has that $\bar{r}_*'=0$. Thus, $\bar{r}$ is initially decreasing. A
computation shows that the following chain of inequalities hold:
\begin{equation}
\label{Inequalities:rbar}
-\alpha<0<2m< r_*<\alpha \leq r_\circledast.
\end{equation}
 Thus, the curve must reach the
singularity ($r=0$) before it can reach the turning point at
$\bar{r}=-\alpha$. It only remains to see whether the conformal
geodesic reaches the singularity in finite amount of proper time. 

\medskip
 A computation using the factorisation
\eqref{FactorisationSchwarzschild} shows that $\bar{r}''=0$
implies the condition
\[
\bar{r}(\bar{r}-\alpha)(\bar{r}+\alpha)+ \bar{r}
(\bar{r}-r_*)(\bar{r}+\alpha) +
\bar{r}(\bar{r}-r_*)(\bar{r}-\alpha) = 2 (\bar{r}-r_*)(\bar{r}-\alpha)(\bar{r}+\alpha).
\]
where it has been assumed that $\bar{r}\neq 0$. A further
rearrangement yields
\begin{equation}
2\bar{r}^2(\bar{r}-r_*) =(\bar{r}^2-\alpha^2)(\bar{r}-2r_*).
\label{SchwarzschildConditionInflexion}
\end{equation}
Using the chain of inequalities in \eqref{Inequalities:rbar} one concludes that
\[
(\bar{r}^2 -\alpha^2)<0, \quad \bar{r}-2r_*<0, \quad \bar{r}-r_*<0.
\]
Thus, the left hand side of condition
\eqref{SchwarzschildConditionInflexion} is negative, while the right
hand side is positive. This shows that there are no values of
$\bar{r}<r_\circledast$ for which $\bar{r}''=0$. Hence, one concludes
that the function $\bar{r}$ reaches the value $\bar{r}=0$ in a finite
value of $\bar{\tau}$ ---that is the conformal curves under
consideration hit the singularity in a finite amount of proper
time. Moreover, it is noticed that $\bar{r}'\rightarrow \infty$ as
$\bar{r}\rightarrow 0$ ---cf. equation \eqref{enemy}.

\medskip
Finally, the conformal geodesic starting at the bifurcation sphere
($r_*=2m$) is covered by the analysis of Section
\ref{Section:CGBifurcationSphere}, by setting $q=0$. One finds the
explicit solution
\[
\bar{r}=m(1+\cos \bar{\tau}).
\]
This conformal geodesic reaches the singularity at $\bar{\tau}=\pi$.


\end{document}